\documentclass[12pt]{article}

\usepackage{color}

\usepackage[margin=3.81cm]{geometry}

\usepackage[utf8]{inputenc}

\usepackage{float}
\usepackage{tikz}
\usepackage{csquotes}
\usetikzlibrary{positioning,arrows}
\tikzset{
m/.style={circle,draw,fill=gray!20,minimum size=5},outer sep=2pt}
\usepackage{setspace}
\usepackage{amsmath,amsfonts,amssymb,amsthm, mathrsfs}

\usepackage{dcolumn}

\newcolumntype{d}[1]{D{.}{.}{#1} } 

\usepackage{verbatim}
\usepackage{bbm}

\usepackage{graphicx,wrapfig, tikz, cancel, hyperref, multirow,rotating}
\usepackage{caption, subcaption}
\usepackage{natbib}

\usepackage{enumerate}
\usepackage[overload]{empheq}
\usepackage{pdflscape} 
\usepackage{hyperref}
\hypersetup{colorlinks,            citecolor=blue,            filecolor=blue,            linkcolor=blue,            urlcolor=blue,            }

\theoremstyle{plain}
\newtheorem{prop}{Proposition}
\newtheorem{lemma}{Lemma}
\newtheorem{remark}{Remark}

\theoremstyle{plain}
\newtheorem{defi}{Definition}
\newtheorem{assumption}{Assumption}

\newcommand{\dd}{\mathrm{d}}

\newcommand{\EE}{\mathbb{E}}

\usepackage{authblk}

\usepackage{abstract}

\begin{document}

\title{\vspace{-2.5cm}Homophily and infections:\\
static and dynamic effects\thanks{%
We thank seminar participants at the University of Antwerp, Universit\`a di Bologna, Oxford University, at the 2nd EAYE Conference at Paris School of Economics, Bocconi University, and Politecnico di Torino. We would also like to thank Daniele Cassese, Satoshi Fukuda, Andrea Galeotti, Melika Liporace, Matthew Jackson, Alessia Melegaro, Luca Merlino, Dunia Lopez--Pintado, Davide Taglialatela, and Fernando Vega--Redondo for useful comments and suggestions. Fabrizio Panebianco acknowledges
funding from the Spanish Ministry of Economía y Competitividad project ECO2107-87245-R. Paolo Pin acknowledges
funding from the Italian Ministry of Education Progetti di Rilevante Interesse
Nazionale (PRIN) grants 2017ELHNNJ, P20228SXNF and 2022389MRW and from Region Tuscany grant Spin.Ge.Vac.S..
A previous version of this manuscript has circulated with title ``Epidemic dynamics with homophily, vaccination choices, and pseudoscience attitudes''.}}

\author{\vspace{-0.5cm} Matteo Bizzarri\thanks{Universit\`a  Federico II and CSEF, Napoli, Italy. 
\href{matteo.bizzarri@unina.it}{matteo.bizzarri@unina.it}},
Fabrizio Panebianco\thanks{Department of Economics and Finance, Universit\`a Cattolica, Milan, Italy. \href{fabrizio.panebianco@unicatt.it}{fabrizio.panebianco@unicatt.it}},
Paolo Pin\thanks{Department of Economics and Statistics, Universit\`a di Siena, 
and BIDSA, Universit\`a  Bocconi, Milan, Italy. \href{paolo.pin@unisi.it}{paolo.pin@unisi.it}}
}


\date{ }

\maketitle


\vspace{-1cm}
\begin{abstract}

We analyze the effect of homophily in  {the diffusion of a harmful state} between two groups of agents that differ in  {immunization} rates.
Homophily has a very different impact on the steady state infection level (that is increasing in homophily when homophily is small, and decreasing when high), and on the cumulative number of infections generated by a deviation from the steady state (that, instead, is decreasing in homophily when homophily is small, and increasing when high).
If immunization rates are endogenous, homophily has the opposite impact on the two groups. 
However, the sign of the group-level impact is the opposite if immunization is motivated by infection risk or by peer pressure. If motivations are group-specific, homophily can be harmful to both groups. 

\end{abstract}

\begin{small}
	\textbf{D85}	Network Formation and Analysis: Theory --
		\textbf{I12}	Health Behavior 

\bigskip

\end{small}

	\section{Introduction}
	
It is well known that diffusion processes in social networks are crucially affected by the homophily between different groups, namely the tendency of members of a group to interact among themselves more than across groups. In this paper, our goal is to show that homophily can have a very different impact on the long-run steady state and on the dynamics of the diffusion.
Specifically, we study the diffusion of a harmful state in a population where immunization is available. We can think of the harmful state as an infectious disease, a harmful behavior (tobacco use, as in \citealp{galeotti2013strategic}\footnote{In the example of tobacco use, immunization \enquote{is
interpreted as a commitment to avoid the temptation of smoking} (\citealp{galeotti2013strategic}).}), or a rumor whose belief may be harmful (as in \citealp{merlino2023debunking}).
In the following, we use for simplicity the terminology of the diffusion of a disease, referring to immunization as vaccination.

We study a stylized environment with two heterogeneous groups of agents, one with high vaccination rate, the other with a low vaccination rate, and we focus on the case of a non-zero stable steady state for infection.  {At time 0 there is a random shock to the number of infected (which we sometimes refer to as an \emph{outbreak}), that generates a dynamic adjustment. After that time, the system converges back to the stable steady state.} 

We analyze what is the impact of homophily on two measures: the steady-state number of infections, and the cumulative number of infected--person--periods generated during by an outbreak (which we refer to as \enquote{cumulative infection}, for simplicity).  {Our main contributions are two: first, we show that the effect of an increase in homophily on the \emph{steady state} infection level can be diametrically opposite to the effect on the cumulative infection.
Second, we show that homophily has an opposite effect on the vaccination rates depending on whether the motivation for infection is a classic \enquote{rational} cost-benefit analysis, or instead depends on peer pressure (that some recent evidence, reviewed in the literature section, highlights as an important factor).}

It is well known that social networks exhibit a high degree of homophily, and that homophily is one of the network characteristics crucially affecting diffusion and contagion (see e.g.~\citealp{jackson2008social}). Moreover, when policy makers have to implement certain policies, such as affecting vaccination uptake, this may  {in turn} affect homophily. For example, there has been a lot of evidence that in the US private and charter schools have a higher level of non-vaccinated children, and this is driven by a larger number of families that use the possibility of religious or philosophical exemptions (see the discussion and references in Section \ref{sec:peer}.). 
 {Another example, focusing on the tobacco use interpretation, is that smokers likely prefer to spend time in spaces that either allow smoking, or have spaces to do it: for example some types of bars, or clubs. Once again, different enforcement of smoking regulations can have the effect of increasing homophily of interaction between people with similar immunization (commitment) rates.}

To model the diffusion of the harmful state, we adopt the SIS (\enquote{susceptible-infected-susceptible}) model, following \cite{galeotti2013strategic} and \cite{jackson2008social}.
The network is stylized and composed of two groups, with heterogeneous vaccination rates.
Vaccination is perfectly effective and is not updated over time. In the basic model, vaccination rates are exogenous. We later endogenize vaccination rates and microfound the discrepancy  {across groups, by assuming that agents in a group perceive a higher cost of vaccination}. With this in mind, we label the two groups \enquote{\textit{vaxxers}} and \enquote{\textit{anti-vaxxers}}.\footnote{ {Actually, our main results of Section \ref{sec:model} (for exogenous vaccination rates) only depend on the presence of heterogeneity between the amount of susceptibles in each group, and not on the fact that this heterogeneity comes from vaccination (even if it is a natural example). This could, for example, represent a difference in infectivity across age or ethnic groups: our result on the impact of homophily would be the same.}} The homophily of contacts between the two groups is modeled by a parameter $h \in [0,1]$, which is the percentage of contacts that people have exclusively with others in their same group, while the rest of contacts is with a fraction of agents drawn at random from the population. For tractability, we focus on the linear approximation of the dynamics around the steady state.

 {In the main analysis, we adopt the SIS model because it is among the most tractable diffusion models, and has often been used in economic theory to describe stylized diffusion processes (\citealp{galeotti2013strategic}, \citealp{jackson2013diffusion} and others). However, in \ref{app:SIR}, we show the robustness  of our main insight in the setting of a SIR model  (\enquote{susceptible-infected-removed}), in which recovered agents are immune, and we characterize how homophily affects the cumulative infection. The main insight, namely that the effect of an increase in homophily on the steady state infection can be very different from the effect on the cumulative infection, generalizes immediately to the SIR  model: indeed, in the SIR model the steady state has zero infections, and so homophily has no effect on the steady state. However, homophily affects the dynamics, and so it naturally has an effect on the cumulative infection. So, the main insight is true also in the SIR model.}


We consider deviations from the steady state of an amount that is stochastic and has zero mean. We study the effect of homophily on two measures of infection: a static one, that is the \textit{steady-state infection}, and one relative to the dynamic, that depends on the size of a deviation from the steady state, and is the discounted sum of each period's infections due to the outbreak, that we call \emph{cumulative infection}. This cumulative measure of infections can be seen as a reduced form of various utilitarian welfare functions that have been used in the recent literature. These examples are all connected to the idea that each infection in each time period has a cost, and those costs must be summed: the more similar being \cite{rowthorn2012optimal}, \cite{farboodi2021internal} and \cite{toxvaerd2022management}. An example in this sense is discussed in detail in Remark \ref{welfare}.\footnote{Moreover, as we discuss in Section \ref{subsec:cum}, this dynamic component can be thought of as the \cite{bonacich1987power} centrality of each group in the network composed by the two groups, where the strength of the connection between groups depends on the amount of probable infections transmitted.} The key characteristic of the cumulative infection is that, contrary to the steady-state level of infection, it contains information on the way the epidemic evolves over time, how long it lasts, and the discount rate.  {If infections are fixed at the steady state, this discounted sum is the same as the steady-state prevalence. However, in an outbreak, when the infection level can vary over time, the two objects can be very different.}

 {In the main result, we do not adopt a specific utility interpretation, because our point is descriptive: our first main contribution is to show that the effect of homophily on infections may be opposite on the steady state and on the cumulative infection.
However, when in Section \ref{sec:preferences} we introduce agents' preferences, to endogenize vaccination choices, we show how a modification of a utilitarian welfare function, with the addition of a health care cost that must be financed through uniform taxes, gives rise to a welfare that depends both on the steady state level of infection \emph{and} the cumulative infection.}\footnote{We thank an anonymous referee for suggesting this interpretation.}


In particular, we find that the steady-state total infection level is increasing for small homophily $h$ and decreasing for large homophily $h$ (Proposition \ref{Steady}). Instead, cumulative infection follows the opposite behavior: it is decreasing in homophily for small $h$ and increasing for $h$ large (Proposition \ref{Cumulative}). 
The key intuition behind the result is as follows. In steady state, a change in homophily has a direct effect of increasing infections in the group with less vaccinated agents, because they meet non-vaccinated people more often, and decrease them in the other, for the symmetric reason. Then, there are indirect effects due to the impact of the steady-state levels on the dynamics. The key determinants of these indirect effects are $(i)$ a (negative) \emph{size} effect: a higher infection decreases susceptibles, hence decreases future infections; and $(ii)$ a (positive) \emph{contagion} effect: higher infection increases future infections, increasing contagion probability. When homophily is small, the size effect is symmetric, hence the sign of the impact of homophily is determined by the contagion effect and is positive; on the other hand, when homophily is large, the contagion effect is symmetric, so the sign is determined by the size effect and is negative.

The effect of homophily on cumulative infections due to an outbreak can also be decomposed in a direct and indirect effect. First, homophily has a direct effect of changing the dynamics, that we explore in Section \ref{subsec:cum} through a two-step example. Namely, the level of homophily affects the sign of the gap between the additional infections generated in the two groups: this acts as a reversal of the direct effect of homophily on steady-state levels, thus reversing the behavior when $h$ is large. Moreover, homophily decreases the convergence rate of the system (Proposition \ref{speed_conv}): intuitively, this suggests that homophily makes the epidemic longer, increasing the number of cumulative infections. Second, there is an indirect effect due to the change in steady-state levels: higher steady-state levels mean that there are less susceptible individuals that can be infected, hence outbreaks are smaller. As a consequence, cumulative infections due to the outbreak are \emph{decreasing} in the steady-state levels.

To study the effect of vaccination rates that adjust when homophily changes, we explore a model in which agents trade-off an heterogeneous vaccination cost with their perceived benefit of vaccination.
Furthermore, we assume that
the two groups differ only in (possibly) size, and in their judgment about the real cost of vaccination, which is deemed higher by anti-vaxxers. This can be thought of as a psychological cost, a sheer mistake, or any phenomenon that may lead to a difference in perceived cost: we remain agnostic on the cause of it as our aim is to study its consequences.\footnote{
In recent years many people either refuse drastically any vaccination scheme or reduce (or delay) the prescribed vaccination.  
	The phenomenon has become more pronounced in the last decades, especially in Western Europe and in the US. See  \cite{larson2016state} for a general cross country comparison,
 \cite{phadke2016association} for the US and
	\cite{funk2017critical} for measles in various European countries.
With the model in the main text we mean to capture not the extremists that would never take a vaccine, but the more general phenomenon of \emph{vaccine hesitancy}, which is more widespread and, so, potentially  more dangerous \citep{trentini2017measles}.
	}

We explore two different possibilities for the motivation of vaccinations: vaccinations motivated by avoiding the risk of infections, and vaccinations motivated by peer pressure. In the former, agents evaluate the benefit of vaccination as the negative of the infection rate: the gain in utility if they do not get the illness. In the latter, agents receive a high benefit from vaccination if many other agents in their neighborhood are vaccinated,  {through a peer effect channel. The cost-benefit model is classic, a variation of the one studied in \cite{galeotti2013strategic}. 
We explore the peer effects model because it is widely believed that peer affect decision making, and specifically decisions related to health and insurance, e.g.: in the use of social insurance (\citealp{markussen2015social}), the adoption of menstrual cups (\citealp{oster2012determinants}), HIV testing (\citealp{godlonton2012peer}), retirement plan choice (\citealp{duflo2003role}). Some recent papers find direct evidence of a role of peer effects in vaccination decisions (see \citealp{hoffmann2019vaccines} for a systematic review). 
In Section \ref{sec:peer} we review further the empirical evidence.}
Real decision making likely involves a mix of peer effects and cost-benefit calculation, so these two cases should be thought of as the two extreme cases in which only one of the two components is visible, for the sake of illustrating the mechanisms. 

Our second main contribution is to show that the effect of homophily on vaccination rates is opposite in the two vaccination models. 
Indeed, if vaccinations are motivated by the risk of infection, an increase in homophily has the effect of increasing risk, hence vaccinations, among anti-vaxxers, and decreasing them among vaxxers (Proposition \ref{prop:groups}). Instead, if vaccinations are motivated by peer pressure the mechanism is the opposite: an increase in homophily increases the peer pressure in the group with more vaccinations (the vaxxer group), hence increasing vaccination among the vaxxers, and decreasing them among the anti-vaxxers (Proposition \ref{prop:peer}). Homophily is the most harmful to vaccinations in a hybrid model in which vaxxers vaccinate according to the risk of infection, while anti-vaxxers according to peer pressure. In such a case homophily unambiguously decreases vaccinations both among vaxxers (because it reduces risk), and among anti-vaxxers (because it decreases the peer pressure).

Moreover, even in the endogenous vaccination models, the basic insight that steady-state and cumulative infection measures may have opposite behavior for the extremes values of homophily survives. In the rational model, this is true when the vaccination cost is high enough (Proposition \ref{prop:endovax}). In the peer-effects based model, it is always true (Proposition \ref{prop:peer}). This happens, essentially, because vaccination rates adjust in opposite directions, hence the additional effect is never too strong.

\bigskip
\subsection*{Related literature}

We contribute to three lines of literature: the literature on epidemics in economics, the literature on contagion and diffusion in networks, and the literature on strategic immunization.

Our contribution to the literature on epidemics in economics is first to study how homophily impacts infections, and more generally to highlight how different risk and time preferences used to evaluate the welfare impact of an epidemic may give different weights to the steady state and to the cumulative infection. Our cumulative measure of infections can be seen as a reduced form of various utilitarian welfare functions that have been used in the recent literature: the most similar being \cite{rowthorn2012optimal}, \cite{farboodi2021internal}, and \cite{toxvaerd2022management}. Other papers use richer models, studying the tradeoffs between economic activity and deaths (both absent from our model): \cite{acemoglu2021optimal}, \cite{brotherhood2021economic}, \cite{bognanni2020economics}. All of these papers do not consider the effect of the social network. \cite{bisin2021efficient} consider health and economic trade-offs, but do not consider dynamics or homophily. 
The structure of the cross-country network is considered in \cite{chandrasekhar2021interacting}, that also consider as the objective of the planner to minimize what they call the \enquote{number of infected--person--periods}, a measure analogous to cumulative infection. None of these papers studies the homophily of interactions.

Our contribution to the literature on contagion is to highlight how the effect of homophily of interactions can be radically different when focusing on the cumulative number of infections over time, rather than the steady state. The closest paper, from which we adapt the basic setting, is \cite{galeotti2013strategic}. That paper was the first to study in depth the effect of homophily on infection rates. However, the authors focus on the steady-state level of infections, rather than the dynamic measure of cumulative infection we focus on. 
It is well known that homophily can facilitate the diffusion of a disease, as illustrated, e.g., in
\cite{jackson2013diffusion}. However, they do not study the impact of homophily on the steady-state levels or the dynamic cumulative infections. \cite{izquierdo2018mixing} and \cite{burgio2022homophily} study the steady state and find a non-monotonic effect of homophily similar to our result, but they do not study the dynamic cumulative infection.
 {A different strand of literature has argued that the whole time evolution of the dynamics is important beyond the steady state, studying departures from the standard mean field approximation, to allow stochastic fluctuations, as in \cite{nakamura2019hamiltonian} and \cite{esen2022generalization}. These papers do not study the effect of homophily.}

Our contribution to the literature on strategic immunization models is to show that the impact of homophily on group-level vaccinations can be opposite if vaccination is motivated by avoiding the risk of infection or by peer pressure.
Our model of vaccinations motivated by infection risk is analogous to \cite{galeotti2013strategic}. 
However, the endogenous vaccination model in \cite{galeotti2013strategic} generates symmetric vaccination in the two groups, because it assumes a homogeneous vaccination cost, while we use heterogeneous vaccination costs precisely to microfound and study different vaccination rates.
 \cite{goyal2015interaction} studies the interaction between the endogenous level of interaction and vaccinations, again in steady state.
The fact that vaxxers tend to vaccinate less when homophily increases is similar to the risk compensation effect studied in \cite{talamas2020free}, that shows that a partially effective vaccination can decrease welfare. Again, our focus is rather on the static-dynamic trade-offs. \cite{chen2014economics} argues that the market mechanism yields inefficiently low levels of vaccination. None of these papers explores vaccination driven by peer pressure.
\footnote{%
There is also a recent literature in applied physics that studies models where the diffusion is simultaneous for the disease and for the vaccination choices. On this, see the review of \cite{wang2015coupled}, and the more recent analysis of \cite{alvarez2017epidemic} and \cite{velasquez2017interacting}. 
} 


	\bigskip
	
		The paper is organized as follows.
		The next section presents the model.
		Section \ref{sec:mechanical} shows  {the main} results for the mechanical model, when all choices are exogenous.
		Section \ref{sec:vaccination} explores the robustness of the results to endogenous vaccination rates.
   {In Section \ref{sec:peer} we extend the model by considering the case in which vaccination choices are motivated by peer pressure.} We conclude in Section \ref{sec:conclusion}.
   {In \ref{app:SIR} we study how our setuo applies to the SIR model.
  In \ref{app:optimal} we discuss the optimal vaccination rates.}
Finally, in \ref{app:proofs} we prove the formal results of our paper.
		

\section{The Model}
	
\label{sec:model}

We consider a simple SIS model with vaccination and with two groups of agents, analogous to the setup in \cite{galeotti2013strategic}. 


Our society is composed of a continuum of agents of mass $1$, exogenously partitioned into two groups.
Agents in each group are characterized by their attitude towards vaccination. In details, following a popular terminology, we label the two groups with $a$, for \emph{anti-vaxxers}, and with $v$, for \emph{vaxxers}. Thus, the set of the two groups is $G:=\{a,v\}$, with $g\in G$ being the generic group. 
Let $q\in[0,1]$ denote the fraction of \emph{anti-vaxxers} in the society and $1-q$ the fraction of \emph{vaxxers}.

People in the two groups meet each other with an \emph{homophilous} bias. We model this by assuming that an agent of any of the two groups has a probability $h$ to meet only someone from her own group and a probability $1-h$ to meet someone else randomly drawn from the whole society.\footnote{$h$ is the \emph{inbreeding homophily} index, as defined in \cite{coleman1958relational}, \cite{marsden1987core}, \cite{mcpherson2001birds} and \cite{currarini2009economic}.
It can be interpreted as the amount of time that agents spend interacting with people in their group, while in the remaining time they meet uniformly at random.  
}
This implies that anti-vaxxers meet each other at a rate of $\tilde{q}_a:=h+(1-h)q$, while vaxxers meet each other at a rate of $\tilde{q}_v:=h+(1-h)(1-q)$. Note that $h$ is the same for both groups, but if $q \ne 1/2$ and $h<1$, then $\tilde{q}_a\neq \tilde{q}_v$.

For each $g\in G$, let $x_g\in[0,1]$ denote the fraction of agents in group $g$ that are vaccinated against our generic disease. It is natural to assume, without loss of generality, that $x_a<x_v$, and by now this is actually the only difference characterizing the two groups. The total number of vaccinated (or average vaccination rate) is $x=qx_a+(1-q)x_v$.
We start by taking $x_a$ and $x_v$ as exogenous parameters, and we endogenize them later.  {We are going to assume $x_g\in(0,1)$ to make the problem nontrivial}.
Similarly, $\rho_g$ and $S_g$ denote, respectively, the fraction of infected and susceptible in group $g$. The total number of vaccinated is denoted $\rho=q\rho_a+(1-q)\rho_v$. Whenever there is possible ambiguity, steady-state variables are denoted with an $SS$ apex, so that the total number of infections in the steady state is denoted $\rho^{SS}$. We omit the $SS$ apex whenever the context makes it clear that we are using steady-state values. Let $\mu$ be the recovery rate of the disease, whereas its infectiveness is normalized to $1$. 

We are going to be concerned with the stable steady state of the system above and with stochastic, zero-mean, deviations from the steady state. At time 0, there is the realization of the random variables $d\rho_0=(d\rho_{0,a},d\rho_{0,v})$, measuring such deviation, where $\EE \dd\rho_{0,a}=\EE \dd\rho_{0,v}=0$.
 For simplicity, we assume that the deviation is symmetric across the two groups: $d\rho_{0,a}=d\rho_{0,v}=\overline{d\rho_0}$. This is already sufficient to show the difference between the impact of homophily in the steady state and in the dynamics, which is our goal; so we stick to this simplifying assumption. 

\subsection{The dynamical system}

Setting the evolution of the epidemic in continuous time, we study the fraction of infected people in each group. 

The differential equations of the system are given by:
\begin{eqnarray}
\label{system1}
\dot{\rho}^a & =F_a(\rho_a,\rho_v)= & S_a \Big( \tilde{q}_a \rho_a + (1-\tilde{q}_a) \rho_v \Big) - \rho_a \mu;
\nonumber \\
\dot{\rho}^v & =F_v(\rho_a,\rho_v)= & S_v \Big(  \tilde{q}_v\rho_v + (1-\tilde{q}_v) \rho_a \Big) - \rho_v \mu  .
\end{eqnarray}
where $S_g=\big(1-\rho^g-x^g \big)\in[0,1]$ are the fraction of agents who are neither vaccinated, nor infected, and thus susceptible of being infected by other infected agents. Moreover, the shares of infected agents met by anti-vaxxers and vaxxers are given by $\tilde{\rho}_a:=\Big( \tilde{q}_a \rho_a + (1-\tilde{q}_a) \rho_v \Big)$ and by $\tilde{\rho}_v:=\Big(  \tilde{q}_v\rho_v + (1-\tilde{q}_v) \rho_a \Big)$, respectively. Finally, $\rho_a \mu$ and $\rho_v \mu$ are the recovered agents in each group. 


 First, in the next proposition we characterize some properties of the steady states, using a classic result in \cite{lajmanovich1976deterministic} (see also Observation 1 in \citealp{galeotti2013strategic}). 

\begin{prop}[{Homophily and endemic disease}]
\label{SS}
The system \eqref{system1} always admits a unique stable steady state. 
For each $h\in[0,1]$, there exists a $\hat{\mu}(h)>0$ such that (i) if  $\mu<\hat{\mu}(h)$, the stable steady state is interior: $\rho^{SS}_a>0$, $\rho^{SS}_v>0$, while there is another (unstable) steady state in $(0,0)$, whereas (ii)  if $\mu\ge \hat{\mu}(h)$, the unique steady state is $\rho^{SS}=(0,0)$, and is globally asymptotically stable.
\end{prop}

The formal passages of the proof are in \ref{app:proofs}, as those of the other results that follow. In all the paper from now on, we focus on the interior steady state, and to ensure that this exists for all $h$ we assume $\mu<\min_h\hat{\mu}(h)=1-x$, denoting $x=qx_a+(1-q)x_v$ the total vaccination rate.



For analytical tractability, in the following we will approximate the dynamics of outbreaks away from the steady state with the linearized dynamics of the deviation from the steady state $\dd\rho_{i,t}=\rho_{i,t}-\rho_i^{SS}$, for $i=a,v$.

\begin{defi}[Dynamic]

We define the functions $\dd\rho_{a,t}$, $\dd\rho_{v,t}$ as the time evolutions that satisfy: 
\begin{align}
    \begin{pmatrix}
    \dot{\dd\rho}_{a,t}\\
    \dot{\dd\rho}_{v,t}
    \end{pmatrix}&=J\left(\begin{array}{c}
        \dd\rho_{t,a}  \\
        \dd{\rho}_{t,v} 
    \end{array}\right)\ ,\  \quad \dd{\rho}_0=\left(\begin{array}{c}
        \dd \rho_{0,a}  \\
        \dd\rho_{0,v }
    \end{array}\right),
\end{align}
where 
\[
J=\left(\begin{array}{cc}
-\tilde{\rho}^{SS}_a-\mu+S^{SS}_a\tilde{q}_a     &  (1-S^{SS}_a)\tilde{q}_a \\
 (1-S_v)\tilde{q}_v    & -\tilde{\rho}^{SS}_v-\mu+S^{SS}_v\tilde{q}_v
\end{array}\right)
\]
is the Jacobian matrix of \eqref{system1} calculated in the steady state, and $(\dd\rho_{0,a},\dd\rho_{0,v})'$ is the initial magnitude of the outbreak. 
\label{approximation}

Moreover, we denote: $\dd\rho_t=q\dd\rho_{a,t}+(1-q)\dd\rho_{v,t}$.
\end{defi}

If $\mu < \hat{\mu}$, the steady state is stable, implying that $J$ has negative diagonal elements. We denote the determinant of $J$ by $|J|$ and note that it is positive. Additionally, $\hat{\mu}(h)$ (explicitly derived in \ref{app:proofs}) is increasing in $h$. This highlights the first important role of $h$ in comparative statics: if $h$ increases, a disease that was previously non-endemic (because $\mu > \hat{\mu}(h)$) might become endemic as $\hat{\mu}(h)$ increases with $h$, reversing the inequality. Indeed, higher homophily counterbalances the negative effect of the recovery rate $\mu$ on the epidemic outbreak (see also the discussion in \citealp{jackson2013diffusion}).

\subsection{Cumulative infection}


 {Suppose that, instead of the number of infected in the steady state, we care about the total number of infected--person--periods generated by the epidemic, as termed by \citealp{chandrasekhar2021interacting}. For example, if each infection has a cost to the health-care system, this aggregate number would correspond to the total monetary cost of healthcare during an epidemic.}

In this section, we precisely define the concept.
A key simplification is that the dynamics $\dd\rho_{t,a}$ and $\dd\rho_{t,v}$ are, by construction, linear in $\dd\rho_{0}$. Therefore, the total cumulative number of infections over time is also linear in $\dd\rho_{0}$.

\begin{defi}[Cumulative infection]

Define $\tilde{CI}$ as the (normalized) cumulative number of infections due to a deviation from the steady state of size $\dd\rho_0=(\overline{d\rho_0},\overline{d\rho_0})$, discounted with discount rate $r$:
    \begin{align*}
\tilde{CI}&:=r\left(q \int_0^{\infty}e^{-r t}(\rho_a^{SS}+\dd{\rho}_{a,t})\dd t+(1-q)\int_0^{\infty}e^{-r t}(\rho_v^{SS}+\dd{\rho}_{v,t})\dd t\right) \nonumber \\
&=\rho^{SS}+r\int_0^{\infty}e^{-r t}\dd{\rho}_t\dd t \label{outbreak}
\end{align*}
Moreover, thanks to the linearity of $\dd\rho_t$, define $CI$ as the coefficient such that:
\begin{align}
\tilde{CI}(\overline{d\rho_0})&=\rho^{SS}+CI\overline{d\rho_0} 
\end{align}   

\end{defi}

 {Since the number of infected in the steady state and the cumulative number of infections due to the deviation from the steady state are distinct quantities, both potentially of interest for policymakers,
in the following we separately analyze the impact of homophily on the two measures: the steady state infection level $\rho$, and the cumulative infection $CI$. The idea is that both these statistics may be relevant, depending on the problem studied, but homophily has a very different impact on each.}

\section{Steady state vs Cumulative infection}
\label{sec:mechanical}

In this section we start analyzing the pure epidemic part of the model, taking the vaccination rates $x_a$ and $x_v$ as exogenous.
Remember that, in this case, the only difference between the two groups is that $x_a<x_v$. In the following, we drop the superscript $SS$ from steady state variables whenever there is no ambiguity, for ease of notation.



\subsection{Homophily in steady state}

First, we explore what is the effect of homophily in the steady state. Homophily has the effect of increasing the social contacts among agents of the same group: hence, an increase in homophily $h$ has the effect of increasing the amount of not vaccinated people that anti-vaxxers interact with, with the result of increasing the steady-state infection level. The opposite effect is true for the vaxxers. What is the balance of these effects? The next proposition answers.

\begin{prop}[Homophily in the steady state]

In the interior steady state the total infection $\rho$  is increasing if homophily is small enough, and decreasing if homophily is high enough.\footnote{For completeness, we can show it has only one maximum under the assumption that $\mu<(1-x)^2/(1-x_a)$.}

\label{Steady}

\end{prop}


The intuition for the results on group-level infection rates is the following.
An increase in homophily has a direct effect, due to the change in the meeting rates across groups; and an indirect effect, due to the effect that a change in the steady states have. The direct effects are caused by the homophily changing the probability of infection:
\begin{align}
    \partial_h \tilde{\rho}_a&=(1-q)\Delta \rho \nonumber;\\
    \partial_h \tilde{\rho}_v&=-q\Delta \rho.
    \label{direct}
\end{align}
They have opposite signs: anti-vaxxers meet more frequently anti-vaxxers hence, ceteris paribus, their probability of infection goes up. For vaxxers the opposite happens.

The indirect effects are due to the impact that each infection level has on the dynamic increments $\dot{\rho}_a=F_a(\rho_a,\rho_v)$ and $\dot{\rho}_v=F_v(\rho_a,\rho_v)$. They can be decomposed as such:
\begin{align}
\dd F_a=&\partial_{\rho_a}(S_a\tilde{\rho}_a)\dd\rho_a+\partial_{\rho_v}(S_a\tilde{\rho}_a)\dd\rho_v \nonumber\\
=&(\underbrace{-\mu}_{\substack{\text{recovery}\\\text{effect}}}-\underbrace{\tilde{\rho_{a}}}_{\substack{\text{size}\\\text{effect}}}+\underbrace{\tilde{q}S_{a}}_{\substack{\text{contagion}\\\text{effect}}})d\rho_{a}+\underbrace{(1-\tilde{q})S_{a}}_{\substack{\text{contagion}\\\text{effect}}}d\rho_{v}  \nonumber;  \\
\dd F_v=&(\underbrace{-\mu}_{\text{recovery}}-\underbrace{\tilde{\rho_{v}}}_{\text{size}}+\underbrace{\tilde{q}_vS_{v}}_{\text{contagion}})d\rho_{a}+\underbrace{(1-\tilde{q}_v)S_{v}}_{\text{contagion}}d\rho_{v}.
\label{eq:effects}
\end{align}

For example, for group $a$, an increase in the steady-state level $\rho_a$ generates, for group $a$: (i) an increase of recovery, (ii) a decrease in the pool of susceptible agents (size effect), and (iii) a increase in the probability of infection (contagion effect). The recovery effect is constant, and symmetric across groups. Since the increment in infection comes from a product of the amount of susceptible agents and of the probability of infection, each of these two effects are respectively proportional to the level of the other (via the Leibniz differentiation rule). The size effect, that is the reduction of the pool of susceptible agents, is proportional to the infection probability $\tilde{\rho}_a$: hence it is stronger for anti-vaxxers. The size effect is always negative. Finally, there is the contagion effect, due to the increase in the probability of meeting an infected person. This is positive, and its magnitude depends on $q$, but for $q=1/2$ it is proportional to the share of susceptible agents, and so the effect is once again stronger for anti-vaxxers group. Considering group $a$, the recovery and the size effect are only present for a variation of the own steady state level $\rho_a$, while the contagion effect is present both for the own steady state $\rho_a$, and for a variation in the steady state of the other group $\rho_v$.

The indirect effects always have the opposite sign compared to the direct effects, making the overall balance of uncertain sign. The results above indicate that the indirect effects are never strong enough to counterbalance the direct effect in group $a$, so $\rho_a$ always increases with $\rho_{0,a}$. In group $v$, however, the derivative can take both signs: it is negative if the direct effect prevails, and positive otherwise. Since cross-group contagion is part of the indirect effect, the direct effect prevails and $\rho_v$ decreases with $\rho_{0,a}$ when $h$ is large, meaning the two groups are almost separated. Conversely, when $h$ is small, the indirect effect may be stronger than the direct one, causing $\rho_v$ to increase with $\rho_{0,a}$. This occurs when the recovery rate $\mu$ is large enough, making the contagion effect a more significant driver of infection.\footnote{Notice that this reasoning holds only when both infection rates do not reach corner solutions: if for example $x_v=1$, so that all the vaxxers are vaccinated, then $\rho_v=0$ is a constant and does not change; in this case the only relevant derivative is:
	\[
	\partial_{h}\rho_{a}=\partial_{h}\rho=-\frac{S_{a}(1-q)\Delta\rho}{J_{11}}>0
	\]
}

Note that also the effect of homophily on total infections stems from a balance of such direct and indirect effects. 
The expressions for the derivative in \eqref{eq:effects} reflect the fact that the variations of the steady state are a combination of the direct effects from \eqref{direct}, weighted by the responses of the dynamics to a variation in the steady state, in such a way to leave the dynamics at rest. 

Summing up, the sign of $\partial_h \rho$ is determined by
\[
S_a(\tilde{\rho}_v+\mu)-S_v(\tilde{\rho}_a+\mu)
\]
which represents the balance of the strengths of the contagion effects (whose magnitude are proportional to $S_a$ and $S_v$) and of the size effects (whose magnitude are $\tilde{\rho}_a$ and $\tilde{\rho}_v$).

When homophily $h$ is low, the infection probabilities are the same $\tilde{\rho}_a\sim \tilde{\rho}_v$, hence the size effects, that are proportional to them, do not matter: the contagion effect, which is positive, dominates and so infections increase in $h$.
Instead, when homophily is high, the amount of susceptible agents are the same in the two groups, $S_a,S_v\to \mu$, hence the contagion effect does not matter, and the result is determined by the size effect, which is negative: hence homophily decreases total infections.\footnote{One might wonder why the population size $q$ has little effect on the result. Note that the marginal changes in the probability of infection (and hence both $\partial_h\rho_a$ and $\partial_h\rho_v$) depend on the fraction of the population in the \emph{other} group: this is the amount of the change in people met for a unit increase in $h$.
The consequence is that when computing the total infection, the population fractions can be collected, because each term is multiplied by $q(1-q)$, and does not matter anymore.}


\subsection{Cumulative infection}
\label{subsec:cum}

We now analyze how results are affected once we explicitly model the infection dynamic. First, we can note that cumulative infection is closely related to Bonacich centrality:
\[
\begin{pmatrix}
    \tilde{CI}_a\\
    \tilde{CI}_v
\end{pmatrix}=r\int_0^{\infty}e^{(-rI+J)t}\dd\rho_0\dd t=(I-1/rJ)^{-1}\dd\rho_0.
\]
We can see that the vector of cumulative infections in the two groups is equal to the Bonacich centrality  {of each group} in the weighted network defined by the Jacobian matrix $J$.
This expression is going to be useful to make calculations with cumulative infection.
The intuition can be better grasped considering the associated discrete time dynamics, that satisfies: 
\begin{equation}
\begin{pmatrix}
    \dd\rho_{a,t+1}\\
    \dd\rho_{v,t+1}
\end{pmatrix}=J\begin{pmatrix}
   d\rho_{a,t}\\
   d\rho_{v,t}
\end{pmatrix}.
\label{discrete}
\end{equation}
In such a case, the cumulative infection is simply \[
\sum_t r^{-t} J^t \dd\rho_0=(I-1/rJ)^{-1}\dd\rho_0.
\]Each step in the time iteration adds a number of infections proportional to the direct and indirect connections in the weighted connection network defined above up to step $t$. The sum of all the total direct and indirect discounted connections amounts to the total cumulative infection over time, and is equal to the Bonacich centrality. 
The continuous time result can be obtained for the step size going to zero.

The next Proposition is the main result of this section, showing that 
$CI$ has opposite behavior with respect to $\rho^{SS}$.

\begin{prop}

The impact of homophily $h$ on $CI$, $\dd_h CI$, is positive when $h$ is low enough, and negative when $h$ is high enough. 

\label{Cumulative}
\end{prop}

What is the reason for this discrepancy?  {The total effect can be decomposed into a direct effect of $h$ on the dynamics, and an indirect effect, due to $h$ affecting also the steady-state levels:}
\[
\dd_h CI=\underbrace{\partial_hCI}_{\text{direct effect}}+\underbrace{\partial_{\rho_{a}}CI\partial_h\rho_a+\partial_{\rho_{v}}CI\partial_h\rho_v}_{\text{indirect effect}}.
\]
 {In the proof of the above Proposition, we clarify that the direct effect $\partial_hCI$ is positive if $h$ high enough, and the indirect effect has a sign that is exactly opposite to the sign of the effect on the steady state in Proposition \ref{Steady}.} 
In the following paragraphs we try to give intuitions for both.

\paragraph{Intuition: direct effect}
\label{direct_effect}
To better understand the intuition behind the direct effect of $h$ on the cumulative infection, we turn again to the approximate discrete dynamics (\ref{discrete}). Let us analyze a simple two-step discrete version of the dynamics. At $t=1$ we have:
\begin{align*}
d\rho_{a,1}&=(-\tilde{\rho_{a}}-\mu+\tilde{q}_aS_{a})\overline{d\rho_0}+(1-\tilde{q}_a)S_{a}\overline{d\rho_0},\\
 \dd\rho_{a,1}&=(-\tilde{\rho_v}-\mu+\tilde{q}_vS_{a})\overline{d\rho_0}+(1-\tilde{q}_v)S_{v}\overline{d\rho_0}.
\end{align*}
The gap with total infection at the steady state is:
\begin{align*}
d\rho_{1}
=&(-\rho^{SS}-\mu+qS_{a}+(1-q)S_{v})\overline{d\rho_0}=(-2\rho^{SS}-\mu+1-x)\overline{d\rho_0},
\end{align*}
and hence we can see that this is \emph{independent} of homophily $h$. The fact that the two groups have an identical initial deviation $\overline{d\rho_0}$ means that only the average effects matter: the average of the contagion effect terms is equal to the average (total) number of susceptible agents, while the average size effect is equal to the total number of infections. Hence, after one period, only population-level statistics matter.

However, the two deviations $d\rho_{a,1}$ and $d\rho_{v,1}$ are not identical. After one period, the gap between groups' infections $d\rho_{1,a}-d\rho_{1,v}=\Delta \dd\rho$ is:
\[
\Delta \dd\rho_1=(\Delta S-h\Delta\rho^{SS})\overline{d\rho_{0}},
\]
which once again, depends on the size and contagion effects previously discussed. Similarly to the derivatives in \eqref{Steady}, both effects are stronger for anti-vaxxers, so if $h$ is low, since the size effect is symmetric, the contagion effect dominates and the gap is positive; the opposite happens when $h$ is high. These effects are analogous to the effects driving the impact of homophily on steady state infections.

To compute total infections at period 2, we can decompose the new delta infection rates as deviations from the average number of infected agents: $d\rho_{a,1}=d\rho_{1}+(1-q)\Delta \dd\rho_1$ and $\dd\rho_{v,1}=d\rho_{1}-q\Delta \dd\rho_1$.
Then, we can express the new total infections at period 2 as:
\begin{align}
d\rho_{2}
&=-\dd\rho_1^2+q(1-q)h(\Delta S-\Delta\rho^{SS})\Delta \dd\rho_1.
\label{eq:time2}
\end{align}
By linearity, we get two additive terms: one derives from the average component $\dd\rho_1$, while the other derives from the deviations, proportional to $\Delta \dd\rho_{1}$. The first term implies analogous calculations as the total infection at time 1, when starting from homogeneous initial deviations: hence it is also independent of homophily, conditional on the steady state infection. 

The second term is the crucial one, containing the effect of homophily $h$ on the total infections at period 2, conditional on steady state values. Once again, we see that it depends on the balance of size ($\Delta \rho^{SS}$) and contagion ($\Delta S$) effects. However, the important part is that the sign of the effect also depends on the increment at the previous period, whose sign depends itself on $h$. In particular, when $h$ is large and the size effect dominates, the gap $\Delta \dd\rho_1$ is \emph{negative}, hence the overall sign is \emph{positive}, which is the opposite conclusion than what we get in the steady state.\footnote{If, instead, $h$ is small, the gap $\Delta \dd\rho_1$ is positive, and the sign of the term $\Delta S-\Delta\rho^{SS}$ is uncertain: for $h=0$ it is positive if and only if $\mu>(1-x)/2$. This is because when $\mu$ is large the contagion effect is more important than past infections.} The reason is that, as Equation \eqref{eq:time2} describes, homophily changes, together with the steady state levels,
also the intermediate steps of the dynamics. 

This example shows a short run intuition for the discrepancy between the static and dynamic effects, that the Proposition above shows formally in infinite time. The next paragraph 
shows how we can get similar long run intuitions analyzing the behavior of the convergence rate, as measured by the smallest eigenvalue (in absolute value).

\paragraph{Intuition: indirect effect}

To understand the intuitive connection between $CI$ and the share of infected agents at steady state, $\rho^{SS}$, it is useful to first focus on the case in which groups are totally separated, namely $h=1$. In this case, each group follows an independent standard SIS equation (we report the equation for the $a$ group):\footnote{An analogous result can be obtained for $h=0$, because in this case we can average the two equations and obtain an equation for the evolution of the total infection $\rho$ directly:
\[
\dot{\rho} =(1-\rho-x)\rho-\mu\rho.
\]}
\[
\dot{\rho}_a=S_a\rho_a-\mu\rho_a =(1-\rho_a-x_a)\rho_a-\mu\rho_a.
\]
The linearization of this process is given by:
\[
\dot{d\rho}_a=-\rho^{SS}_ad\rho_a \Longrightarrow \dd\rho_a(t)=e^{-\rho^{SS}_at}\rho_{0}
\]
 so that we can analytically compute: $CI=\frac{r\dd\rho_{0}}{\rho^{SS}_a+r}$: the cumulative infection is inversely proportional to the steady state infection. The intuition is that the higher the steady-state infection, the fewer susceptible agents are, so that the deviation from the steady state is smaller and the system goes back to the steady state faster.
 
 If $h \neq 1$, the dynamics is paired and a clear analytical inverse proportionality is lost. However, to clarify the dynamic intuition behind the mechanism, in the next paragraph we show that, following a similar intuition, the convergence rate of the dynamics, measured by the smallest eigenvalue of $J$, is decreasing in the steady-state levels.

\paragraph*{Convergence rate}
\label{convtime}

To formalize the intuitions discussed in the previous two paragraphs, we consider the following classic definition (see, e.g. \citealp{gabaix2016dynamics}).

\begin{defi}[Convergence rate]

The convergence rate of the system after a deviation of size $\dd\rho_0$ is
\[
CR=-\lim_{t \to \infty}\frac{\log\lVert e^{tA}\dd\rho_{0}\rVert}{t}.
\]
\end{defi}

A classic property of linear systems is that the speed of convergence can be measured by eigenvalues: we show formally that this is also the case here. Moreover, we formally show that the speed of convergence is decreasing in both steady-state levels, $\rho^{SS}_{a}$ and $\rho^{SS}_{v}$. This provides further insight behind the mechanism of the indirect effect of Proposition \ref{Cumulative}. Further, we show that, similarly to \cite{golub2012homophily}, also in this context homophily decreases the convergence rate, at least when $h$ and $\mu$ are large: this also provides a long-run intuition behind the direct effect in Proposition \ref{Cumulative}. 

\begin{prop}

The speed of convergence is equal to the absolute value of the eigenvalue of smallest modulus of the matrix $J-r I$:
\[
CR=\lambda_2
\]

When $h\to 1$, $CR$ is increasing in $h$. When $h \to 0$, $CR$ is increasing if and only if $\mu>\frac{1-x}{2}$. Moreover, $CR$ is decreasing in both $\rho^{SS}_a$ and $\rho^{SS}_v$.

\label{speed_conv}
\end{prop}


\section{Vaccination choices: infection risk vs peer effects}
\label{sec:vaccination}

 {Our model so far is agnostic on agents' preferences. However, in reality vaccination choices are taken by agents, and 
the fraction of vaccinated agents in the population itself might depend on homophily when taking into account that vaccination is endogenous. So, it might be important to explore how agents decide to vaccinate or not. Endogenizing vaccination choices requires us to carefully define what are the motivations for vaccinating. Because of the motivations described in the introduction, we explore two alternatives: first, a standard model (a variant of \citealp{galeotti2013strategic}) in which agents correctly anticipate the equilibrium probability of infection; and another one, in which agents' choices are driven by peer effects.}

\subsection{Vaccination based on rational choices}

\subsubsection{Preferences}
\label{sec:preferences}

 {Agents might vaccinate paying a cost, or not vaccinate incurring the risk of becoming infected. We assume that agents take vaccination decisions ex-ante, before an epidemic actually takes place, and cannot update their decision during the diffusion.}\footnote{In general, agents might find it optimal to postpone the vaccination, for two reasons: to learn first the size of the shock away from the steady state $d\rho_0$, or to postpone in the future the disutility of paying the one-off vaccination cost, $-re^{-rt}c$. However, if we allow full flexibility to agents, the model becomes intractable, because now agents vaccinate at different times, and we have to keep track of two additional differential equations for the evolution of $x_a(t)$ and $x_v(t)$. So, for the sake of tractability, we simplify the problem by assuming that agents can vaccinate only before the epidemic starts.
} Vaccination is perfectly effective.
In a descriptive spirit, we microfound the discrepancy in vaccination rates assuming that anti-vaxxers have a cost larger than vaxxers of a uniform amount $d$. However, our focus being on the effect of homophily on infections, we do not dig deeper into the motivations for this different cost evaluation. Thus, we assume that, for vaxxers, vaccination costs are $c^v\sim U[0,1/k]$, whereas for anti-vaxxers $c^a\sim U[d/k,1/k+d/k]$. $k$ is a parameter reflecting the distribution of vaccination costs in the population: a high $k$ means that vaccination costs are generally small, whereas a low $k$ means that they are high.\footnote{
Our model would not change dramatically if we attribute the difference in perception to the costs of becoming sick,
but we stick to the first interpretation because it makes the computations cleaner.}

 {The benefits of vaccinations are weighed against the disutility of infection. Focus on an agent $i$ in group $a$, denoted as $i \in a$. We define its \emph{infected status} $I_{i,a,t} \in \{0,1\}$ as 1 if they are infected, and zero otherwise, and analogously for $i \in v$. The expected utility of agent $i$ in group $g$ is:}
\[
U_{i,g}=\begin{cases}
    -r\mathbb{E} \int_0^{+\infty}I_{i,g,t}dt & \text{if } i \text{ non vaccinated}\\
    -c_i & \text{o.w.}
\end{cases}
\]
Since $I_{i,g,t}$ is bounded and has finite expectation, we can pass  the expectation inside the integral, so that:
\[
\mathbb{E} \int_0^{+\infty}e^{-rt}I_{i,g,t}dt= \int_0^{+\infty}e^{-rt}\mathbb{E}(I_{i,g,t})dt
\]

 {Conditional on the realization of the initial deviation $d\rho_0$, the variables $(I_{i,a,t}, I_{i,v,t})$ follow a continuous time Markov chain. Following \cite{pastor2015epidemic}, the time evolution of the expectations $\mathbb{E}(I_{i,g,t}\mid d\rho_0)$ satisfies:}
\begin{align}
\dfrac{d}{dt}\mathbb{E}(I_{i,a,t}\mid d\rho_0)&=(1-\mathbb{E}(I_{i,a,t}\mid d\rho_0))\tilde{\rho}_{a,t}-\mu \mathbb{E}(I_{i,a,t}\mid d\rho_0)\\
\dfrac{d}{dt}\mathbb{E}(I_{i,v,t}\mid d\rho_0)&=(1-\mathbb{E}(I_{i,v,t}\mid d\rho_0))\tilde{\rho}_{v,t}-\mu \mathbb{E}(I_{i,v,t}\mid d\rho_0)
\end{align}

This is exactly the same equation satisfied by $\left(\dfrac{\rho_{a,t}}{1-x_a},\dfrac{\rho_{v,t}}{1-x_v}\right)$, and by the uniqueness of the solution, it must be:
\[
\mathbb{E}(I_{i,a,t}\mid d\rho_0)=\dfrac{\rho_{a,t}}{1-x_a},\quad 
\mathbb{E}(I_{i,v,t}\mid d\rho_0)=\dfrac{\rho_{v,t}}{1-x_v}
\]
Furthermore, using our linearization, we have $\EE \rho_{a,t}=\rho_a^{SS}$ and $\EE \rho_{v,t}=\rho_v^{SS}$. So, the ex-ante expected utilities are:
\[
U_{i,g}=\begin{cases}
    -\dfrac{\rho_g^{SS}}{1-x_g} & \text{if } i \text{ non vaccinated}\\
    -c_i & \text{o.w.}
\end{cases}
\]
 {Since there is a continuum of agents, each individual takes the fraction of vaccinated in the population as given.
As in \cite{galeotti2013strategic}, when deciding to vaccinate or not, agents compare the fraction of time spent in the infected state and the cost of vaccination. The difference is that we assume that the vaccination cost is heterogeneous across agents.
So, an agent $i$ in group $a$ vaccinates if and only if $c_i<\rho_a$. The fraction of agents that vaccinate in group $a$ is thus $x_a=k \dfrac{\rho_a}{1-x_a}-d$, provided the solution is interior. Similarly, $x_v=k\dfrac{\rho_v}{1-x_v}$ for vaxxers.}

\begin{remark}
\label{welfare}
\textbf{Utilitarian welfare}

With the individual preferences just defined, the utilitarian welfare would be:
\begin{align*} 
W&=q\int U_{i,a} di+(1-q)\int U_{i,v} di\\
&=-q\left(\int^{x_a/k}_{d/k}kcd+\int_{x_a/k}^{(1+d)/k}k\dfrac{\rho_a}{1-x_a}dc\right)\\
&-(1-q)\left(\int^{x_v/k}_0kcdc+\int_{x_v/k}^{1/k}k\dfrac{\rho_v}{1-x_v}dc\right) \\
&=-q\dfrac{1}{2k}\left(x_a^2-d^2\right)+\dfrac{d\rho_a}{1-x_a} -(1-q)\dfrac{x_v^2}{2k} -q\rho_a^{SS}-(1-q)\rho_v^{SS}
\end{align*}
Now, in the decentralized equilibrium $x_a=k\dfrac{\rho_a}{1-x_a}-d$, and we can use this to rewrite the above as:
\[
W=-\dfrac{1}{2}\left(q(x_a+d)^2+(1-q)x_v^2\right) -\rho^{SS}
\]

\textbf{Utilitarian welfare with health care costs}

 {Now, suppose that each infection above the steady-state level costs $b'$ to the health care system. This can be interpreted as follows: the health care system is prepared to deal with a number of infected individuals equal to the steady-state $\rho^{SS}$, but additional infections have a cost in terms of congestion for the health care system:}\footnote{We thank an anonymous referee for suggesting this interpretation.}  {for example because they require building up additional capacity in terms of personnel, equipment, and hospital space. This is what happened, for example, during the COVID-19 pandemic. These additional health care costs have to be financed through taxes. Therefore, the government must raise an amount equal to $b'\mathbb{E}(\overline{d\rho_{0}} \mid \overline{d\rho_{0}} \ge 0){CI}$, and it raises this amount via a lump-sum tax $t$ levied uniformly on all agents: so this would not affect the vaccination decision of individuals. 
So, the aggregate welfare would become:}
\begin{align}
W^{\text{congestion}}=-\dfrac{1}{2}\left(q(x_a+d)^2+(1-q)x_v^2\right)-\rho^{SS}-bCI
\label{welfare_congestion}
\end{align}
where we defined $b:=b'\EE( \overline{d\rho_{0}}\mid \overline{d\rho_{0}}\ge 0)$, that is positive by definition. 

This shows a simple rationale for which total welfare depends on both the steady-state infection levels, and the cumulative infection.
    
\end{remark}

The following lemma guarantees the existence and uniqueness of an interior vaccination equilibrium.
\begin{lemma}
\label{lemma:choices}
If $d$ and $k$ are small enough (the detailed conditions are in the proof in the Appendix), the equations:
\begin{align}
\begin{cases}
x_a&=k \dfrac{\rho^{SS}_a}{1-x_a}-d\\
x_v&=k\dfrac{\rho_v^{SS}}{1-x_v}
\end{cases}
\end{align}
define a unique equilibrium $(x_a^*,x_v^*)$, whenever the solution is interior, and $x_a^*\le x_v^*$.

\end{lemma}
The proof of Lemma \ref{lemma:choices} uses a version of the global implicit function theorem.

\subsubsection{The effect of homophily on vaccination rates}

Now we study what is the effect of a change in homophily on the equilibrium values of $x_a^*$ and $x_v^*$.

\begin{prop}

In the model with endogenous vaccination, with interior solutions for vaccination rates, if $\mu<(1-x)^2/(1-x_a)$, homophily has the opposite effect on the vaccination rates of the two groups: $x_a^*$ is increasing in $h$, and $x_v^*$ is decreasing in $h$.

    \label{prop:groups}
\end{prop}

 The mechanism behind the comparative statics is that, as $h$ increases, the group with more vaccinated people (the vaxxers) is more protected against infection, so the expected cost of infection $k\rho_v$ decreases, and as a result, a smaller fraction of the vaxxers is vaccinated: $x_v$ is decreasing in $h$. The opposite happens for anti-vaxxers.

\subsubsection{The effect of homophily on steady state and cumulative infection}

We have seen that homophily has opposite impacts on the fraction of vaccinated agents in the two groups. We now study the balance of these effects on the aggregate infection level. The next Proposition shows that, at least for $k$ and $d$ small enough, the signs of the derivatives with respect to homophily are the same as in the model with exogenous vaccination rates. This indicates that the insight that homophily might have a very different impact on the steady state or throughout the dynamics is robust and does not disappear once we endogenize vaccination rates.

\begin{prop}

If $d$ and $k$ are small enough, we have that:

\begin{enumerate}

    \item in the limit for $h\to 0$ the steady state infection $\rho$  is increasing, and in the limit for $h\to 1$ is decreasing;

 \item in the limit for $h\to 0$ the cumulative infection $CI$  is decreasing, and in the limit for $h\to 1$ is increasing.
    
\end{enumerate}

\label{prop:endovax}
\end{prop}

\subsection{Vaccination motivated by peer pressure}
\label{sec:peer}

In this section we explore an alternative scenario in which vaccination decisions are not driven by a correct evaluation of the infection risk. Instead, we consider agents who do not know or consider the infection risk at all but vaccinate purely motivated by peer pressure: they vaccinate if a sufficient fraction of the agents they meet are vaccinated. Under this vaccination model, we show that if homophily is large enough, no one vaccinates, and so infection rates are higher with respect to the case in which vaccination is based on infection risk. Moreover, the behavior of vaccination rates with respect to homophily is the opposite of what happens with vaccination based on the risk of infection: $x_a$ is \emph{decreasing} in $h$, while $x_v$ is increasing in $h$. This is because in this case agents are insensitive to risk, and increasing homophily tends to increase the peer pressure to vaccinate for vaxxers, and to decrease it for anti-vaxxers.
Despite these differences, we show that in this alternative framework, the steady-state and the cumulative infection display the same qualitative behavior, with respect to a change in homophily for $h \to 1$, as in the baseline framework we presented.

\subsubsection{Why Peer Effects in Vaccinations?}

 {Recently, many papers have found evidence of peer effects in vaccination choices. \cite{raosocial} use randomized assignment of students to dorms to estimate the peer effect in flu vaccination decisions, finding a significant effect. \cite{hoffmann2019vaccines} study randomized allocation of worker shifts to weekdays or weekends and find evidence of peer effects on vaccination decisions, which they argue are driven by the desire to conform to the social norm of the group. In an RCT, \cite{sato2019peer} find peer effects in tetanus vaccination take-up in Nigeria. \cite{ibuka2018analysis} empirically analyze online survey data, finding evidence of positive peer effects in vaccination decisions.}

 {More generally, there is indirect evidence suggesting that health behavior displays clustering in social networks: \cite{christakis2007spread} (obesity), \cite{sorensen2006social} (health plan choice). A phenomenon that attracted lot of attention is the fact that vaccination exemptions required by parents on religious grounds tend to concentrate in some types of schools, such as Steiner and Waldorf schools, where they constitute a large fraction of the population. The phenomenon attracted a lot of discussion by medical researchers: \cite{ernst2012implications}, \cite{may2003clustering}, \cite{muscat2011gets}, \cite{zier2020attend}.
It is documented for California in \cite{silverman2019lessons}.
Recent evidence shows that similar trends \href{https://www.ilfoglio.it/salute/2019/05/23/news/chiuso-un-altro-asilo-no-vax-la-mappa-delle-scuole-fantasma-che-piacciono-agli-antivaccinisti-256645/}{happened in Italy} and \href{https://www.nytimes.com/2019/04/15/nyregion/measles-nyc-yeshiva-closing.html}{have been considered a cause of a measles outbreak in Manhattan in April 2019} (on this, see \citealp{mashinini2020impact}, and the earlier discussion by \citealp{shaw2014united}).
There is reason to believe that these more permissive schools have attracted parents that are more skeptical of vaccinations. For example, \cite{sobo2015social} argues that school community norms have an important impact in vaccine skepticism among families of children attending Steiner schools.}

 {\cite{edge2019observational} document that vaccination patterns in a network of social contacts of physicians in Manchester hospitals are correlated with being close in the network. Geographic clustering of vaccination behavior is another fact that suggests a role of peer effects: \cite{lieu2015geographic} show that vaccine-hesitant people are more likely to communicate with each other than with others. \cite{cullen2023direct} find that health care providers can affect the vaccination decisions of the social networks of their patients.}

 {Another strand of literature studies individual decision-making related to vaccinations, showing that behavior often departs from a simple cost-benefit calculation. In many cases, providing more information does not change the minds of vaccine-hesitant people (see \citealp{nyhan2013hazards,nyhan2014effective} and \citealp{nyhan2015does}).}


\subsubsection{Formal model}

 {
The agents in the two groups have heterogeneous vaccination costs, whose distribution follows the same assumptions as in Section \ref{sec:vaccination}. The key difference is the evaluation of the health effect of being vaccinated. Rather than evaluating the disutility connected to the average infection rate, as in Section \ref{sec:vaccination}, agents evaluate the peer pressure stemming from other agents' vaccination decisions. In particular, agents compare the vaccination cost with a peer pressure term equal to the average number of infected individuals in their neighborhood.
}

 {Formally, an agent in group $a$ with heterogeneous cost $c$ vaccinates if $c<\tilde{q}_a x_a+(1-\tilde{q}_a)x_v$, and an agent in group $v$, analogously, vaccinates if $c<(1-\tilde{q}_v)x_a+\tilde{q}_v x_v$. Hence, the vaccination rates at interior solutions are defined by:}
\begin{align}
x_a=&k(\tilde{q}_a x_a+(1-\tilde{q}_a)x_v)-d,\nonumber \\
x_v=&k(\tilde{q}_v x_v+(1-\tilde{q}_v)x_a).
\label{vax_peer}
\end{align}

When $h$ is not too large, these equations have interior solutions given by:
\begin{align*}
x_a=&\frac{d (1-k (1-(1-h) q))}{(k-1) (1-h k)}, \quad
x_v= \frac{d (1-h) k q}{(k-1) (1-hk)}.
\end{align*}
 {while for $h$ high enough the peer effects are so high that either no one vaccinates or everyone vaccinates. In particular, this means that, for $h\to 1$, we have $\partial_h x_a=\partial_h x_v=0$. So, for $h\to 1$, the effect of $h$ when vaccinations are exogenous is exactly the same as when vaccinations are endogenous, and so the results of Section \ref{sec:mechanical} apply.}
  
The derivatives of the interior solutions with respect to homophily are:
\begin{align*}
\partial_h x_a&=-\frac{d k (1-q)}{(h k-1)^2}, \quad 
\partial_h x_v=\frac{d k q}{(h k-1)^2},
\end{align*}
 {
where it  is immediately clear that the signs are opposite with respect to the signs of the derivatives in the rational model in Section \ref{sec:vaccination}. The reason is that, here, homophily magnifies the influence of peer effects within each group. In the $a$ group, where there are fewer vaccinated individuals, an increase in homophily tends to decrease peer pressure to vaccinate, thereby reducing infection levels. This is diametrically opposite to the effect of homophily in the rational model: there, homophily magnifies the risk of infection from individuals of the same group, leading to an increase in vaccinations in the $a$ group. The opposite occurs in the $v$ group.
}

 {The next Proposition characterizes all the vaccination equilibria, and formalizes the results discussed above.}
\begin{prop}

In the model of vaccination induced by peer pressure:
\begin{enumerate}
    
\item No vaccinations ($x_a^*=x_v^*=0$) is always an equilibrium. Moreover, there is also an interior equilibrium if $k>1$, $h<1-\frac{k-1}{kq}$ and $d<\frac{(k-1)(1-hk)}{k q(1-h)}$ (the complete characterization can be found in the proof in Appendix \ref{app:peer}).

\item In the range of $h$ such that the vaccination rates are interior, $\partial_h x_a<0$ and $\partial_h x_v>0$.

\item For $h\to 1$, we have $d_h CI>0$, and $d_h \rho^{SS}<0$, as in Proposition \ref{prop:endovax}.

\end{enumerate}

\label{prop:peer}
    
\end{prop}


\subsubsection{A mixed model}

 {
Finally, as a worst-case scenario from the planner's perspective, we consider a model with mixed motivations. Specifically, we consider the case where anti-vaxxers, instead of correctly evaluating the risk, rely on peer pressure, whereas vaxxers are more prone to evaluate risks accurately. This example is consistent with empirical evidence showing that anti-vaccine attitudes are more prone to diffusion in online social networks (see, e.g., \citealp{puri2020social}). Under these assumptions, the effect of homophily can be the most negative. Indeed, for some parameter values, an increase in homophily unambiguously decreases vaccination rates in \emph{both groups}: vaxxers vaccinate less due to perceived lower risk, while anti-vaxxers vaccinate less due to peer pressure. To be formal, assume:
}
\[
\begin{cases}
x_a&=k(\tilde{q}_a x_a+(1-\tilde{q}_a)x_v)-d,\\    
x_v&=k\rho_v.
\end{cases}
\]

Indeed, thanks to the implicit function theorem the derivatives of vaccination rates are the following:
\begin{align*}
\left(\begin{array}{cc}
\partial_h x_a  \\
     \partial_h x_v
\end{array}\right)=-\frac{1}{det}\left(\begin{array}{cc}
     -(1-k\partial_{x_v}\rho_v)k(1-h)\Delta x-k^2(1-\tilde{q}_a) \partial_h \rho_v \\
-k^2\partial_{x_a}\rho_v(1-h)\Delta x-(1-k\tilde{q}_a)k \partial_h \rho_v
\end{array}\right)
\end{align*}
where $det=(1-k\partial_{x_v}\rho_v)(1-k\tilde{q}_a)-k^2(1-\tilde{q}_a)\partial_{x_a}\rho_v$. If $1>k\tilde{q}_a$, since $\partial_{x_v}\rho_v<0$ and $\partial_{x_a}\rho_v<0$ we have $det>0$. Moreover, since $\partial_h \rho_v<0$ it follows $     \partial_h x_v<0$. Finally, if $\Delta x$ is small enough, we also have $\partial_h x_a <0$.

Hence, we conclude that when the group with higher cost evaluation peer pressure is the driver of vaccination, while in the other group is the fear of infection, for some parameter values homophily unambiguously decreases equilibrium vaccinations in both.

\section{Conclusion}
\label{sec:conclusion}

	

 {
The problem of vaccine skepticism is a complex one that requires analysis from multiple perspectives, including psychological, medical, and social dimensions. Our study investigates the effects of homophily on infection dynamics within a population divided into "vaxxers" and "anti-vaxxers." Our key findings can be summarized as follows.
}

 {
Firstly, we demonstrate that homophily has a nuanced impact on infection dynamics. Specifically, homophily shows, intuitively, a hump-shaped impact on steady-state infection levels but a U-shaped impact on the cumulative number of infections. This finding underscores the importance of considering both static and dynamic effects when evaluating the role of homophily in disease spread.
}

 {
Secondly, we explore the differences between endogenous and exogenous vaccination rates. When vaccination rates are endogenous, homophily affects the two groups differently. For anti-vaxxers, increased homophily reduces peer pressure to vaccinate, leading to lower infection levels within this group. Conversely, for vaxxers, homophily enhances vaccination rates due to increased peer pressure, thereby mitigating overall infection rates through higher vaccination levels.
}

 {
Thirdly, we investigate scenarios where vaccination decisions are influenced by peer pressure rather than a rational evaluation of infection risk. Our results indicate that high levels of homophily can result in lower vaccination rates in both groups. Vaxxers may perceive a lower risk, while anti-vaxxers are swayed by peer pressure, both contributing to reduced vaccination uptake.
}

 {
Lastly, while our findings provide valuable insights into the interplay between homophily and vaccination behavior, they also highlight potential challenges in formulating effective public health strategies. For instance, policies aimed at reducing contact with anti-vaxxers might inadvertently prolong outbreaks and exacerbate infection rates, depending on the time preferences of the policy planner. However, this point requires careful consideration of broader epidemiological and social factors.
}

 {
By addressing these points, we hope to provide a clearer understanding of the nuanced effects of homophily on vaccination behavior and infection dynamics, which can inform more effective public health strategies.
}

\clearpage

	\appendix
	\global\long\def\thesection{Appendix \Alph{section}}
	\global\long\def\thesubsection{\Alph{section}.\arabic{subsection}}
	\setcounter{equation}{0} \global\long\def\theequation{\alph{equation}}
	\setcounter{lemma}{0} \global\long\def\thelemma{\Alph{lemma}}
	\setcounter{prop}{0} \global\long\def\theprop{\Alph{prop}}
	\setcounter{defi}{0} \global\long\def\thedefinition{\Alph{defi}}
	\setcounter{figure}{0} \global\long\def\thefigure{\Alph{figure}}

\section{Analysis of the SIR model}
\label{app:SIR}

In this section, we study a diffusion described by a SIR model, and we characterize how the cumulative infection depends on homophily. We do so, because the SIR model is the workhorse model used to describe viral diseases, as for example Covid-19, and so it increases the range of applicability of our results, beyond the settings well-described by the SIS model.

In the SIR model, the steady state infection level is always zero, making it independent of homophily. However, the cumulative infection depends on the development of the dynamic, and is affected by homophily: so, our main insight about the difference between the steady state and the cumulative infection translates to the SIR model.
In the following, we give more structure to this statement, characterizing analytically the effect of homophily on the cumulative infection, under an assumption on the parameter $\mu$ that makes the problem more tractable. We note that the result does not rely on any linear approximation of the dynamics around the steady state.\footnote{Indeed, this approximation is not valid for the SIR model, since the steady states are not hyperbolic, that is they have zero eigenvalues, and so the linearization theorem fails (see \citealp{brauer2012mathematical}, and \citealp{meiss2007differential}).} 

First, we define the model. The key assumption is that recovered agents are not susceptible to the disease, but immune. So, they behave as agents that got a vaccination. We are going to denote the sum of recovered and vaccinated agents in the two groups as, respectively, $R_a$ and $R_v$. In this section, we limit ourselves to the case of exogenous vaccination, and we keep the assumption that the vaccination rates do not change over time. So, vaccination rates $x_a$ and $x_v$ represent the initial levels of removed agents: further agents become removed once they get infected and then recover.
The equation describing the SIR model are the following.
\begin{align}
    \dot{\rho}_{a}&=\tilde{\rho}_{a}(1-R_{a}-\rho_{a})-\mu\rho_{a}\nonumber\\
    \dot{R}_{a}&=\mu\rho_{a}\nonumber\\
    \dot{\rho}_{v}&=\tilde{\rho}_{v}(1-R_{v}-\rho_{v})-\mu\rho_{v}\nonumber\\
    \dot{R}_{v}&=\mu\rho_{v}
    \label{SIR}
\end{align}
with the initial conditions: $R_a(0)=x_a$, $R_v(0)=x_v$, $\rho_a(0)=\rho_{v,0}=\overline{d\rho}_0$.

First of all, we note that in this case, the cumulative infection is the same as the number of recovered in the long-run steady state:
\[
CI=q\int_0^{\infty}\rho_{a,t}dt+(1-q)\int_0^{\infty}\rho_{v,t}dt=\dfrac{1}{\mu}(qR_a^{SS}+(1-q)R_v^{SS})=\dfrac{1}{\mu}R
\]
where $R=qR_a^{SS}+(1-q)R_v^{SS}$ denotes the total number of recovered in the steady  state.
Since the long-run steady state is zero, here there is no need of discounting.

The formal result is as follows. We make the simplifying assumption $\mu>1-x_a$ because it makes some steps of the proof analytically tractable.\footnote{Note that for here we do not need to make any assumption to ensure an interior stable steady state, since in every steady state infection is zero, for any $\mu$.} We note that the result does not rely on any linear approximation of the dynamics around the steady state.

\begin{prop}

In the SIR model \eqref{SIR}, suppose that $x_v>x_a$ and $\mu>1-x_a$. Then 
$CI$ is increasing in $h$ if and only if $R_v^{SS}>R_a^{SS}$.

In particular, this condition is satisfied if $h$ is close to 0, or $h$ is close to 1 and $\mu$ large enough.
    
\end{prop}

\begin{proof}

To derive the result, we use the approach of \cite{pastor2001epidemic}, that allows to write an implicit system of equations that defines the number of recovered at the end of the epidemic.

First, let us solve the equation for the susceptible agents:

\begin{align*}
\dot{S}_{a}	&=-\tilde{\rho}_{a}S_{a}\\
S_{a}(0)	&=1-x_{a}-\overline{d\rho}_0
\end{align*}
In the calculations from now on, we assume $\overline{d\rho}_0\to 0$, to simplify the notation. This is inconsequential (in the formulas for $CI$, it is sufficient to redefine the vaccination rates).

Integrating this equation, we get: $S_{a}(t)=(1-x_{a})e^{-\int_{0}^{t}\tilde{\rho}_{a}(s)ds}$.

Now, solving the equation for $\dot{R}_a$, we get: $R_{a}=x_{a}+\mu\int_{0}^{t}\rho_{a}(s)ds$, so we can write:
\begin{align*}
\int_{0}^{t}\tilde{\rho}_{a}(s)ds	&=\tilde{q}_{a}\int_{0}^{t}\rho_{a}(s)ds+(1-\tilde{q}_{a})\int_{0}^{t}\rho_{v}(s)ds
\\
&=\dfrac{1}{\mu}\left(\tilde{q}_{a}\left(R_{a}(t)-x_{a}\right)+(1-\tilde{q}_{a})\left(R_{v}(t)-x_{v}\right)\right)\\
&	=\dfrac{1}{\mu}\left(\tilde{R}_{a}(t)-\tilde{x}_{a}\right)    
\end{align*}
where: 
\begin{align*}
    \tilde{R}_{a}(t)&=\tilde{q}_{a}R_{a}(s)+(1-\tilde{q}_{a})R_{v}(s)
\end{align*}

Differentiating:
\begin{align*}
\dfrac{d\tilde{R}_{a}(t)}{dt}	&=\tilde{\rho}_{a}(t)\\
	&=\tilde{q}_{a}(1-R_{a}-S_{a})+(1-\tilde{q}_{a})(1-R_{v}-S_{v})\\
	&=1-\tilde{R}_{a}(t)-\tilde{q}_{a}(1-x_{a})e^{-\dfrac{1}{\mu}\left(\tilde{R}_{a}(t)-\tilde{x}_{a}\right) }-(1-\tilde{q}_{a})(1-x_{v})e^{-\dfrac{1}{\mu}\left(\tilde{R}_{v}(t)-\tilde{x}_{v}\right) }
 \end{align*}
and analogously for $\tilde{R}_v$.

Now, we use the fact that $CI_{a}=\dfrac{1}{\mu}(R_{a}-x_{a})$, $CI_{v}=\dfrac{1}{\mu}(R_{v}-x_{v})$, and defining $\tilde{CI}_{a}=\tilde{q}_{a}CI_{a}+(1-\tilde{q}_{a})CI_{v}$, $\tilde{CI}_{v}=(1-\tilde{q}_{v})CI_{a}+\tilde{q}_{v}CI_{v}$.
In the long run we know that $\tilde{\rho}_{a}(t)\to0$, so we get:
\begin{align*}
    \tilde{R}_{a}&	=1-\tilde{q}_{a}(1-x_{a})e^{-\tilde{CI}_{a}}-(1-\tilde{q}_{a})(1-x_{v})e^{-\tilde{CI}_{v}}\\
\tilde{R}_{v}&	=1-\tilde{q}_{v}(1-x_{v})e^{-\tilde{CI}_{v}}-(1-\tilde{q}_{v})(1-x_{a})e^{-\tilde{CI}_{a}}
\end{align*}

Now:
\[\begin{pmatrix}\tilde{R}_{a}\\
\tilde{R}_{v}
\end{pmatrix}=Q\begin{pmatrix}R_{a}\\
R_{v}
\end{pmatrix}
\]
defining 
\[Q=\begin{pmatrix}\tilde{q}_{a} & 1-\tilde{q}_{a}\\
1-\tilde{q}_{v} & \tilde{q}_{v}
\end{pmatrix},
\]
the adjacency matrix of the contact network.

So we can rewrite the above  formula as:
\[
\begin{pmatrix}R_{a}\\
R_{v}
\end{pmatrix}	=Q^{-1}\left(\boldsymbol{1}-Q\begin{pmatrix}(1-x_{a})e^{-\tilde{CI}_{a}}\\
(1-x_{v})e^{-\tilde{CI}_{v}}
\end{pmatrix}\right)
	=\left(\boldsymbol{1}-\begin{pmatrix}(1-x_{a})e^{-\tilde{CI}_{a}}\\
(1-x_{v})e^{-\tilde{CI}_{v}}
\end{pmatrix}\right)
\]
where $\boldsymbol{1}$ is a vector of ones, and
because $Q^{-1}\boldsymbol{1}=\boldsymbol{1}$, since all powers remain row-stochastic. This approach significantly simplifies the analysis.

So, we get equations that define implicitly the cumulative infections:
\begin{align*}
    F_{a}&	=-CI_{a}+\dfrac{(1-x_{a})}{\mu}\left(1-e^{-\tilde{CI}_{a}}\right)=0\\
F_{v}&	=-CI_{v}+\dfrac{(1-x_{v})}{\mu}\left(1-e^{-\tilde{CI}_{v}}\right)=0
\end{align*}

Now let us compute the derivative.
\begin{align*}
    J_{CI}F	&=-I+\dfrac{1}{\mu}\begin{pmatrix}(1-x_{a})e^{-\tilde{CI}_{a}}\tilde{q}_{a} & (1-x_{a})e^{-\tilde{CI}_{a}}(1-\tilde{q}_{a})\\
(1-x_{v})e^{-\tilde{CI}_{v}}(1-\tilde{q}_{v}) & (1-x_{v})e^{-\tilde{CI}_{v}}\tilde{q}_{v}
\end{pmatrix}\\
	=&-I+EQ,
\end{align*}
defining 
\[
E=\begin{pmatrix}\dfrac{1}{\mu}(1-x_{a})e^{-\tilde{CI}_{a}} & 0\\
0 & \dfrac{1}{\mu}(1-x_{v})e^{-\tilde{CI}_{v}}
\end{pmatrix}.
\]

Further defining $\Delta CI=CI_a-CI_v$, we get:
\begin{align*}
    J_{h}F &	=-\dfrac{1}{\mu}\begin{pmatrix}(1-x_{a})e^{-\tilde{CI}_{a}}\left((1-q)(-\Delta CI)\right)\\
(1-x_{v})e^{-\tilde{CI}_{v}}q(\Delta CI)
\end{pmatrix}\\
	&=E\begin{pmatrix}1-q\\
-q
\end{pmatrix}\Delta CI=E\begin{pmatrix}1-q\\
-q
\end{pmatrix}\Delta CI
\end{align*}

So, the object of interest is:
\[
\partial_{h}CI	=(I- EQ)^{-1}E\begin{pmatrix}1-q\\
-q
\end{pmatrix}\Delta CI
	=( E^{-1}-Q)^{-1}\begin{pmatrix}1-q\\
-q
\end{pmatrix}\Delta CI
\]
The first expression highlights better the connection with the centrality, the second is more compact. By construction, the inverse has positive elements, and its eigenvalues have positive real part, so the determinant is positive. The key is the trade-off of the two direct effects. Under our assumption on $\mu$, the matrix is an $M$-matrix (that is, with non-positive off-diagonal entries and eigenvalues with nonnegative real parts), so it is invertible, and has positive determinant.

By a direct calculation we find:
\[
\partial_{h}CI=\mu\dfrac{q(1-q)}{det(E^{-1}- Q)} \left(\dfrac{e^{\tilde{CI}_{v}}}{1-x_{v}}-\dfrac{e^{\tilde{CI}_{a}}}{1-x_{a}}\right)\Delta CI
\]

From the steady state equations we get: 
\[
\dfrac{e^{\tilde{R}_{v}}}{1-x_{v}}-\dfrac{e^{\tilde{R}_{a}}}{1-x_{a}}=\dfrac{1}{1-R_{v}}-\dfrac{1}{1-R_{a}}
\] 
that is positive if $R_{v}>R_{a}$. In the next paragraph, we prove that $\Delta CI>0$. So, the derivative of cumulative infection with respect to homophily is positive if and only if $R_{v}>R_{a}$.

\textbf{Proof that $\Delta CI>0$.}

What is the sign of $\Delta CI$? Focus on $h\to0$. In this case, $CI_{a}-CI_{v}=\dfrac{(x_{v}-x_{a})}{\mu}\left(1-e^{-CI}\right)>0$. So, for $h\to 0$, $CI_a>CI_v$.

We can analytically compute the derivatives of the cumulative  infection for each group, and we find:
\[
\partial_hCI_a=\dfrac{\mu^2(1-q)(1-E_1^{-1})}{det((\mu E^{-1}-Q)^{-1})}\Delta CI\quad \partial_hCI_v=-\dfrac{\mu^2q(1-E_2^{-1})}{det( E^{-1}-Q)}\Delta CI
\]
where $E_1$ and $E_2$ are the diagonal elements of the matrix $E$, and they are both smaller than 1.
So we get:
\[
\partial_h \Delta CI=\dfrac{\mu((1-q)(1-E_1^{-1})+q(1-E_2^{-1}))}{det( E^{-1}-Q)}\Delta CI
\]
The sign of this expression is the same of $\Delta CI$. So, at least for $h\to 0$, $\partial_h \Delta CI>0$.

Now suppose by contradiction that for some $h$ we have $\Delta CI<0$. Since $\partial_h \Delta CI(0)>0$, by continuity, there must be a value $\overline{h}>h$ such that $\partial_h \Delta CI(\overline{h})=0$ and $\partial_h \Delta CI(h)>0$ for $h<\overline{h}$. Moreover, we can express $\Delta CI(\overline{h})$ as:
\[
\Delta CI(\overline{h})=\Delta CI(0)+\int_0^{\overline{h}}\partial_h \Delta CI(s)ds>0
\]
because $\Delta CI(0)>0$, and under the above assumption the integral is positive. By the expression for the derivative of $\Delta CI$, since $\Delta CI(\overline{h})>0$, it follows that $\partial_h \Delta CI(\overline{h})>0$, contradicting the assumption that $\partial_h \Delta CI(\overline{h})=0$. So, the derivative must be positive for all $h$, and so it must be that for all $h$ we have $\Delta CI>0$.

\textbf{When is $R_v>R_a$ satisfied?}

To check when this condition is satisfied, consider the limit for $h\to 0$. In this limit, $\tilde{CI}_a=\tilde{CI}_v=CI$. Subtracting the equations for the recovered, we get: $R_{v}-R_{a}=\Delta x(1-e^{-CI})>0$, so the cumulative infection is increasing.

For $h\to 1$ instead, the recovered satisfy:
\begin{align*}
R_{a}&	=1-(1-x_{a})e^{-CI_{a}}\\
R_{v}&	=1-(1-x_{v})e^{-CI_{v}}    
\end{align*}

Dividing the equations:
\[
\dfrac{1-x_{a}}{1-x_{v}}=\dfrac{1-R_{a}}{1-R_{v}}e^{CI_{a}-CI_{v}}=\dfrac{1-R_{a}}{1-R_{v}}e^{\dfrac{1}{\mu}(R_{a}-R_{v})}e^{\dfrac{1}{\mu}(x_{v}-x_{a})}
\]
Now if $\mu$ grows very large, the exponential term go to 1, and so the recovered become close to the vaccination rates: so $R_v>R_a$.

\end{proof}

\section{Optimal Vaccination Rates}
\label{app:optimal}

 {In this section, we explore the optimal vaccination rates that the planner would choose to optimize welfare as defined in Section \ref{sec:preferences}, and we compare with the decentralized equilibrium of the vaccination model of Section \ref{sec:vaccination}.}

 {In our model, there are two sources of inefficiency. The first is classic and arises because agents, when vaccinating, do not internalize the benefit of their decision on the decreased risk of infection for other agents. The second is the anti-vaxxer bias, which leads them to overestimate the vaccination costs by $d$.}
 {The key simplification is that, from the perspective of the planner, the two groups of people are ex-ante identical, so the planner would treat society as one unique group. As a consequence, the only variable to decide is the total level of vaccination $x$. Moreover, from the perspective of the planner the disutility from expected infection is equal to $\dfrac{\rho^{SS}}{1-x}$ for both groups.}

 {Since agents are heterogeneous, one question is which subset of agents it would be optimal to vaccinate. Agents are heterogeneous only in vaccination costs, while the disutility of infection depends only on the total fraction of vaccinated individuals $x$. Thus, it follows immediately that for a given fraction of vaccinated individuals $x$, the optimal policy is to vaccinate all agents with lower vaccination costs, up to the threshold that generates $x$ total vaccinated. This threshold is $x/k$. Given this discussion, the expression for welfare in this context becomes:}
\begin{align}
W=&-\int^{x/k}_0kcdc-\int_{x/k}^{1/k}k\dfrac{\rho^{ss}}{1-x}dc\nonumber\\
&=-\dfrac{1}{2}x^{2}/k-\rho^{ss}
\label{welfare_vaccination}
\end{align}

 {If the groups are homogeneous, the dynamics simplify significantly, and the steady state becomes:}
\[
\rho^{SS}=\begin{cases}
    1-x-\mu & x\le 1-\mu\\
    0 & x> 1-\mu
\end{cases} 
\]

 {Thus, the efficient vaccination level $x$ is the maximizer of the expression \eqref{welfare_vaccination}, taking into account how the steady state depends on $x$. It is immediate to check that the expression for welfare is concave, and the unique solution is:}
\[
x^*=\min\{k,1-\mu\}
\]
This is because if $x^*>1-\mu$, the steady state has zero infections, and it is not optimal to vaccinate agents beyond that point.

We collect these results in the following Proposition.

\begin{prop}
The efficient vaccination level satisfies:
\[
x^*=\min\{k,1-\mu\}
\]

Moreover, the optimum is larger than the vaccination level in the decentralized equilibrium, i.e.: $x^*> x$. 
 
\label{prop:optimal}
    
\end{prop}

\begin{proof}

The welfare is concave. The FOC for $x<1-\mu$ gives $x^*=k$. For $x>1-\mu$, the derivative is always negative and so $x^*$ cannot be higher than $1-\mu$. So, the optimum must be $\min\{k,1-\mu\}$.

In the decentralized equilibrium with rational agents, we would have:
\[
x=k\dfrac{\rho}{1-x}=k\dfrac{1-x-\mu}{1-x}<k,
\]
and so we conclude $x^*>x$. 
    
\end{proof}

\section{Proofs}
\label{app:proofs}

\subsection{Proofs of Section \ref{sec:model}}
\subsubsection*{Proof of Proposition \ref{SS} (page \pageref{SS})}

\begin{proof} The Jacobian of the dynamical system is:
\[
J=\left(
\begin{array}{cc}
 -\mu -\tilde{\rho}_a+\tilde{q}_a S_a & (1-\tilde{q}_a)S_a \\
 (1-\tilde{q}_v)S_v & -\mu -\tilde{\rho}_v+\tilde{q}_v S_v 
\end{array}
\right)=\left(
\begin{array}{cc}
 A & B \\
 C & D
\end{array}
\right)
\]
In the following, we are denoting the entries of the Jacobian as $A,B,C,D$, for simplicity.

Define:
\begin{align*}
\lambda_1&=-\frac{1}{2}\left(A+D-\sqrt{(A-D)^{2}+4 B C}\right)\\
\lambda_2&=-\frac{1}{2}\left(A+D+\sqrt{(A-D)^{2}+4 B C}\right)
\end{align*}
It can be checked that $-\lambda_1$ and $-\lambda_2$ are the two eigenvalues of the matrix $J$ (we define the quantities with the minus because of convenience in later formulas).
Since $B C\ge0$, it follows that the eigenvalues are real and distinct, and from the expression it follows $\lambda_1\ge \lambda_2$. 

Now let us focus on the Jacobian computed in the point $(0,0)$:
\[
J^0=\left(
\begin{array}{cc}
 -\mu +\tilde{q}_a (1-x_a) & (1-\tilde{q}_a)(1-x_a) \\
 (1-\tilde{q}_v)(1-x_v) & -\mu +\tilde{q}_v (1-x_v) 
\end{array}
\right)
\]
and denote the corresponding eigenvalues as $-\lambda^0_1$ and $-\lambda^0_2$.

By Theorem 3.1 in \cite{lajmanovich1976deterministic}, for the system \eqref{system1} there are two possibilities: if $\lambda_2^0\ge 0$ (so that both the eigenvalues are negative), then there in only the zero steady state, and it is asymptotically stable. If instead $\lambda^0_2<0$, then there is also a unique interior steady state, and is asymptotically stable for any initial condition $(\rho_a,\rho_v) \neq (0,0)$.

Using the expression above, the condition $\lambda_2^0< 0$ is equivalent to $\mu<\hat{\mu}(h):=\frac{1}{2}\left(T+\Delta\right)\in[0,1]$, where $T:=\tilde{q}_a(1-x_a)+\tilde{q}_v(1-x_v)$ and $\Delta:=\sqrt{T^2-4h(1-x_a)(1-x_v)}$. Notice that this is increasing in $h$. For this reason, $\min_h \hat{\mu}=\hat{\mu}(0)=1-x$. 
\end{proof}

\subsection{Proofs of Section \ref{sec:mechanical}}

\subsubsection*{Proof of Proposition \ref{Steady} (page \pageref{Steady})}

We need first to establish some useful relations.

\begin{lemma}
Suppose $x_a\le x_v$. In the interior steady state, we have $\rho_a\ge \tilde{\rho}_a\ge \tilde{\rho}_v\ge\rho_v$, $S_a>S_v$, and $x_v-x_a\ge \rho_a-\rho_v$. Moreover $\tilde{\rho}_a-\tilde{\rho}_v=h(\rho_a-\rho_v)$, and the diagonal elements of the Jacobian, $A$ and $D$, are negative.

If $h\to 1$, we have $\rho_a^{SS}\to 1-x_a-\mu$, $\rho_v^{SS}\to 1-x_v-\mu$, $\hat{\mu}\to 1-x_a$, $S_a\to \mu$ and $S_v\to \mu$.

If $h\to 0$, we have $\rho_a^{SS}\to \frac{(1-x_a)(1-x-\mu)}{1-x}$, $\rho_v^{SS}\to \frac{(1-x_v)(1-x-\mu)}{1-x}$, and $\hat{\mu} \to 1-x$, where $x=x_aq_a+x_v(1-q_a)$ is the average number of vaccinated. Moreover $S_a\to \frac{1-x_a}{1-x}\mu$, $S_v\to \frac{1-x_v}{1-x}\mu$.
\label{relations}
\end{lemma}

\begin{proof}
In the interior steady state:
\[
\frac{S_a}{S_v}=\frac{(1-\tilde{q}_v)\frac{\rho_a}{\rho_v}+\tilde{q}_v}{(1-\tilde{q}_a)\frac{\rho_v}{\rho_a}+\tilde{q}_a}
\]
If $\rho_a<\rho_v$ and $x_a<x_v$, then $S_a>S_v$, but the fraction above implies $S_v>S_a$, which is a contradiction. Hence $\rho_a>\rho_v$. From the identity $\tilde{\rho}_a-\tilde{\rho}_v=h(\rho_a-\rho_v)$ it follows $\tilde{\rho}_a>\tilde{\rho}_v$, and since they are averages $\rho_a\ge \tilde{\rho}_a\ge \tilde{\rho}_v\ge\rho_v$. Finally:
\[
S_{a}>S_{a}\frac{\tilde{\rho}_{a}}{\rho_{a}}=S_{v}\frac{\tilde{\rho}_{v}}{\rho_{v}}>S_{v}
\]
The limits for $h\to 0$ and $h\to 1$ follow from solving the system \ref{system1} in the respective cases, and from continuity of the solution of a differential equation with respect to the parameters.

Finally, if $\mu<\hat{\mu}(h)$, the diagonal elements of $J$, computed in the interior steady state, are both negative. We know from stability that $A+D=-\lambda_1-\lambda_2<0$, so at least one of them must be negative; moreover since $B$ and $C$ are nonnegative, and the determinant must be positive: $AD-BC>0 \implies AD>0$: so both $A$ and $D$ are negative.
\end{proof}

\bigskip

Now we are ready to prove Proposition \ref{Steady}.

\begin{proof}

Using the implicit function theorem we can compute the derivatives of infection rates in the steady state:
\begin{align*}
\partial_{h}F_{a}&	=(1-q)S_{a}\Delta\rho\\    
\partial_{h}F_{v}&	=-qS_{v}\Delta\rho
\end{align*}

\begin{align*}
\partial_{h}\rho_{a}
 &=-\frac{S_{a}(1-q)\Delta\rho}{|J|}(S_{v}-\tilde{\rho}_{v}-\mu)\\
\partial_{h}\rho_{v}
&= \frac{S_{v}q\Delta\rho}{|J|}(S_{a}-\tilde{\rho}_{a}-\mu)
\end{align*}

Moreover, $S_{v}-\tilde{\rho}_{v}-\mu=S_{v}-\tilde{\rho}_{v}-S_v\frac{\tilde{\rho}_v}{\rho_v}<0$ since $\tilde{\rho}_v>\rho_v$, so it follows that $\rho_a$ is always increasing in $h$. From the steady state equation $(1-x_a-\rho_a)\tilde{\rho}_a/\rho_a=\mu$ it follows that also $\tilde{\rho}_a$ is increasing. 

Concerning the derivative of $\rho_v$, if $h \to 0$, then 
$S_{a}-\tilde{\rho}_{a}-\mu=\mu\frac{1-x_a}{1-x}-(1-x)$. This is negative if and only if $\mu<(1-x)^2/(1-x_a)$. Similarly, if $h \to 1$ $S_{a}-\tilde{\rho}_{a}-\mu=S_a-\rho_a-S_a\frac{\tilde{\rho}_a}{\rho_a}=S_a-\rho_a-S_a<0$. Moreover, by the previous conclusions on $\rho_a$ and $\tilde{\rho}_a$, we get that $\partial_h \rho_v$ can only have one zero (because $S_a-\tilde{\rho}_a$ is decreasing): so $\rho_v$ is either decreasing or hump-shaped with one maximum. 
From the steady state equation $(1-x_v-\rho_v)\tilde{\rho}_v/\rho_v=\mu$ again it follows that $\tilde{\rho}_v$ is increasing if and only if $\rho_v$ is.

The total infection is:
\begin{align*}
\partial_{h}\rho	&=\frac{q(1-q)\Delta\rho}{|J|}(-S_{a}(S_{v}-\tilde{\rho}_{v}-\mu)+S_{v}(S_{a}-\tilde{\rho}_{a}-\mu))\\
&=\frac{q(1-q)\Delta\rho}{|J|}(S_{a}(\tilde{\rho}_{v}+\mu)-S_{v}(\tilde{\rho}_{a}+\mu))
\end{align*}
	
If $h=0$ we get: $\partial_{h}\rho\propto(\Delta x-\Delta\rho)(\mu+\rho)>0$,
while for $h=1$ we get (since $S_{a}=S_{v}=\mu$):

\[
\partial_{h}\rho\propto-\Delta\rho\mu<0
\]
Moreover, if $\mu<(1-x)^2/(1-x_a)$, the derivative is monotonically decreasing, so the total infection is concave. 

Moreover, in the limit for $h\to 1$:
\begin{align*}
\partial_{h}\rho_{a} &\to -\frac{\mu(1-q)\Delta x}{\rho_a\rho_v}(-\rho_{v})\\
\partial_{h}\rho_{v}
&= \frac{\mu q\Delta x}{\rho_a\rho_v}(-\rho_{a})
\end{align*}
In the limit for $h\to 0$:
\begin{align*}
\partial_{h}\rho_{a}
 &=-\frac{S_{a}(1-q)\Delta\rho}{|J|}(S_{v}-\rho-\mu)\\
\partial_{h}\rho_{v}
&= \frac{S_{v}q\Delta\rho}{|J|}(S_{a}-\rho-\mu)
\end{align*}

The behavior as a function of the vaccination rates is:
\begin{align*}
\partial_{x_{a}}F_{a}&	=-\tilde{\rho}_{a} &  \partial_{x_{a}}F_{v}&	=0\\
\partial_{x_{a}}\rho_{a}&	=-\frac{1}{|J|}(-J_{22}\tilde{\rho}_{a})<0 & \partial_{x_{a}}\rho_{v}&	=-\frac{1}{|J|}(J_{21}\tilde{\rho}_{a})<0\\
\end{align*}
so, all infection rates are decreasing in $x_a$, 
and analogously for the derivative with respect to $x_v$.
\end{proof}

\subsubsection*{Proof of Proposition \ref{Cumulative} (page \pageref{Cumulative})}

 {We need the following Lemma.}

\begin{lemma}

Call the vector of group-level cumulative infection values $\vec{CI}=(CI_a,CI_v)$.
The derivative of the cumulative infection with respect to a parameter $y$ can be expressed as:
\[
\partial_y \vec{CI}=(q,1-q)'(r I-J)^{-1}\partial_{y} J(r I-J)^{-1}d\rho_0
\]

\end{lemma}
\begin{proof}
The vector of the cumulative infections $\vec{CI}=(CI_a,CI_v)$ is the solution of:
\[
r \vec{CI}=d\rho_0+J\vec{CI}
\]
where $\partial_{y} J$ is the element by element derivative of $J$.
Differentiating it and solving we obtain:
\[
\partial_{y}\vec{CI}=(r I-J)^{-1}\partial_{y} J(r I-J)^{-1}d\rho_0
\]

Since $r I-J$ is an \emph{M-matrix}  {(that is, with non-positive off-diagonal entries and eigenvalues with nonnegative real parts)}, the inverse has positive elements.   
\end{proof}

Now we can prove the Proposition.
\begin{proof}
For simplicity, write $d\rho_0=d\rho_{0,a}=d\rho_{0,v}$.
The total derivative is $d_h CI=\partial_h CI+\partial_{\rho_a}CI\partial_h \rho_a+\partial_{\rho_v}CI\partial_h \rho_v$.

 Now, since:
\[
\partial_{h}J=\left(\begin{array}{cc}
-(1-q)\Delta\rho & 0\\
0 & q\Delta\rho
\end{array}\right)+\left(\begin{array}{cc}
S_{a} & 0\\
0 & S_{v}
\end{array}\right)\left(\begin{array}{cc}
1-q & -(1-q)\\
-q & q
\end{array}\right)
\]
using the Lemma we can explicitly calculate the derivatives for $h=0$ or $h=1$, and using $\mu < 1-x$ we obtain the following.
\begin{align*}
&\partial_{h}CI^{\text{out}}\mid_{h=1}
=-\frac{(1-q)q\rho_{0}(x_{a}-x_{v})^{2}(-2r+\mu+x_{a}+x_{v}-2)}{(-r+\mu+x_{a}-1)^{2}(-r+\mu+x_{v}-1)^{2}}>0\\
&\partial_{h}CI^{\text{out}}\mid_{h=0}  \\=
&-\frac{\mu  (1-q) q \dd\rho_0 (x_a-x_v)^2 (2 \mu +q
   (x_a-x_v)+x_v-1)}{(1-x)^2 (-r -1+x) ( 1-x-\mu+r)^2}>0\\
   \iff & \mu> \frac{1-x}{2}
\end{align*}
because the denominator $-r +q
   (x_a-x_v)+x_v-1$ is negative.

Moreover, the derivative of the Jacobian matrix with respect to infection rates is:
\[
\partial_{\rho_{a}}J=\left(\begin{array}{cc}
-2\tilde{q}_{a} & -(1-\tilde{q_{a}})\\
0 & -(1-\tilde{q}_{v})
\end{array}\right) \quad \partial_{\rho_{v}}J=\left(\begin{array}{cc}
-(1-\tilde{q}_{a}) & 0\\
-(1-\tilde{q}_{v}) & -2\tilde{q}_{v}
\end{array}\right)
\]so we conclude that the derivative of \emph{both} cumulative infections with respect to infection rates are negative. In particular, for $h=0$ or $h=1$ we can explicitly write the derivatives of the total CI as:
\begin{align*}
\partial_{\rho_a} CI\mid_{h=1}=-((2 q \dd\rho_0)/(-1 + x_a - r + \mu)^2)\rho_{0}\\
\partial_{\rho_v} CI\mid_{h=1}=-(2 (1 -   q) \dd\rho_0)/(-1 + x_v - r + \mu)^2\rho_{0}\\
\partial_{\rho_a} CI\mid_{h=0}=-\frac{2 q \dd\rho_0}{(-r +\mu +q (x_a-x_v)+x_v-1)^2}\\
\partial_{\rho_v} CI\mid_{h=0}=-\frac{2 (1-q)
   \dd\rho_0}{(-r +\mu +q (x_a-x_v)+x_v-1)^2}
\end{align*}
We conclude that, for $h\to 0$, $(\partial_{\rho_a} CI,
\partial_{\rho_v} CI) \propto -(q,1-q)$, and so:
\[
\partial_{\rho_a} CI\partial_{h} \rho_a\mid_{h=0}+\partial_{\rho_v} CI\partial_{h} \rho_a\mid_{h=0} \propto -\partial_h \rho,
\] 
so the indirect effect is negative.
The direct effect, instead, depends on $\mu$: for $\mu< \frac{1-x}{2}$ they have the same sign, and the derivative is negative. For $\mu> \frac{1-x}{2}$, using the results above we can compute the total effect, that is: 
\begin{align*}
\partial_{h}CI^{\text{out}}\mid_{h=0}+\partial_{\rho_a} CI\partial_{h} \rho_a\mid_{h=0}+\partial_{\rho_v} CI\partial_{h} \rho_a\mid_{h=0}\\
=\frac{\mu  (1-q) q \dd\rho_0 \Delta x^2 (-2 r +2 \mu -3(1-x)}{(1-x)^2 (r +(1-x))
	(-r +\mu -(1-x))^2}
\end{align*}
that has the same sign as $-2 r +2 \mu -3(1-x)<-2 r -(1-x)<0$, since $\mu<1-x$. So the indirect effect dominates.

Instead, for $h$ high, since $\dfrac{\partial\rho_{v}}{\partial h}<0$, we know that the term:
\[
-\dfrac{2d\rho_{0}}{(1-x_{v}-\mu+r)^{2}}\dfrac{\partial\rho_{v}}{\partial h}
\]
is positive. So, it must be larger than $-\dfrac{2d\rho_{0}}{(1-x_{a}-\mu+r)^{2}}\dfrac{\partial\rho_{v}}{\partial h}$,
because the denominator is smaller. So it is true that:
\begin{align*}\dfrac{\partial CI}{\partial\rho_{a}}\dfrac{\partial\rho_{a}}{\partial h}+\dfrac{\partial CI}{\partial\rho_{v}}\dfrac{\partial\rho_{v}}{\partial h}&	=-\dfrac{2qd\rho_{0}}{(1-x_{a}-\mu+r)^{2}}\dfrac{\partial\rho_{a}}{\partial h}-\dfrac{2(1-q)d\rho_{0}}{(1-x_{v}-\mu+r)^{2}}\dfrac{\partial\rho_{v}}{\partial h}\\
&	>-\dfrac{2d\rho_{0}}{(1-x_{a}-\mu+r)}\left(q\dfrac{\partial\rho_{a}}{\partial h}+(1-q)\dfrac{\partial\rho_{v}}{\partial h}\right)\\
&=-\dfrac{2d\rho_{0}}{(1-x_{a}-\mu+r)^{2}}\dfrac{\partial\rho}{\partial h}>0
\end{align*}
because we know that $\dfrac{\partial\rho}{\partial h}<0$ for $h\to1$. So, for $h$ high both the direct and indirect effects are positive.
\end{proof}

\subsubsection*{Proof of Proposition \ref{speed_conv} (page \pageref{speed_conv})}

\begin{proof}
	Using the results in \cite{bernstein1993some}, we can express the exponential matrix as a function of eigenvalues, and directly compute the limit:
\[
-\lim_{t \to \infty}\frac{\log\lVert e^{tA}\dd\rho_{0}\rVert}{t}=\lambda_{2}-\lim_{t \to \infty}\frac{\log\lVert\left(\frac{\lambda_{1}-\lambda_{2}e^{-(\lambda_{1}-\lambda_{2})t}}{\lambda_{1}-\lambda_{2}}I+\frac{1-e^{(\lambda_{1}-\lambda_{2})t}}{\lambda_{1}-\lambda_{2}}A\right)\dd\rho_{0}\rVert}{t}
\]
\[
=\lambda_{2}-\lim_{t \to \infty}\frac{\log\lVert(\lambda_{1}-\lambda_{2}e^{-(\lambda_{1}-\lambda_{2})t})\dd\rho_{0}+(1-e^{-(\lambda_{1}-\lambda_{2})t})A\dd\rho_{0}\rVert}{t}=\lambda_{2}
\]

Concerning the behavior of $\lambda_2$ as a function of $h$, we use a standard result on eigenvalue perturbations (\citealp{demmel1997applied}, Th 4.4). Namely, if $\lambda$ is a simple eigenvalue of $J-r I$:
\[
\partial_h\lambda=\frac{v'\partial (J-r I)u}{v'u}
\]
where $v$ is the left and $u$ the right eigenvector relative to $\lambda$.

In our case both eigenvalues are simple and we can explicitly solve for the eigenvectors. For $\lambda=-\lambda_2$, we obtain:
\begin{align*}
u&	=\left(\frac{2B}{\sqrt{(A-D)^{2}+4BC}-(A-D)},1\right)\\
v&	=\left(\frac{2C}{\sqrt{(A-D)^{2}+4BC}-(A-D)},1\right)
\end{align*} 
and they have both positive components, so $v'u>0$: hence the sign of the derivative is determined by the numerator. By results of the previous proposition, $\partial_{\rho_a} J$ and $\partial_{\rho_v} J$ have both negative elements, hence the derivatives $\partial_{\rho_a}\lambda_2$ and $\partial_{\rho_a}\lambda_2$ are negative.

Moreover, when $h \to 1$, both eigenvectors converge to $(0,1)$. So in the limit of $h\to 1$ we get:
\[
\partial_h \lambda_2=-\partial D=-q\Delta\rho-qS_{v}<0
\]
so that for high enough $h$, the speed of convergence is increasing in $h$.
Instead, when $h \to 0$, the eigenvectors converge to $(q,1-q)$ and $(1-x_a,1-x_v)$, and the derivative becomes:
\[
\partial_h \lambda \mid_{h\to 0}=-\frac{(1-q) q (1-x-2\mu) \Delta x^2}{(1-x)^2}
\]
that is positive if and only if $\mu>\frac{1-x}{2}$.
\end{proof}

\subsection{Proofs of Section \ref{sec:vaccination}}

For the proofs of this section, we are going to define $\theta_a=\dfrac{\rho_a}{1-x_a}$ and $\theta_v=\dfrac{\rho_v}{1-x_v}$.
    We can rewrite the system \eqref{system1} with this parameterization as:
\begin{align}
    \dot{\theta}_{a} &=(1-\theta_{a})\tilde{\rho}_{a}-\mu\theta_{a} \nonumber\\
\dot{\theta}_{v}&	=(1-\theta_{v})\tilde{\rho}_{v}-\mu\theta_{v}
\label{system2}
\end{align}

Moreover, we are going to need the following Lemma.

\begin{lemma}

System \eqref{system2} has the same stability properties of System \eqref{system1}.

The derivatives $\partial_{x_a} \theta_a$, $\partial_{x_v} \theta_a$, $\partial_{x_v} \theta_v$, $\partial_{x_v} \theta_a$ are all negative.

    \label{negative_theta}
\end{lemma}

\begin{proof}

The two systems are related by a change of coordinates: $(\rho_{a,t}, \rho_{v,t})=\Theta(\theta_{a,t},\theta_{v,t})$, where the function $\Theta$ maps $(\theta_{a,t},\theta_{v,t})$ into $((1-x_a)\theta_{a,t},(1-x_v)\theta_{v,t})$. Since it is differentiable has differentiable inverse for $x_a,x_v\in (0,1)$, the two systems are \emph{smoothly equivalent}, and so they have the same steady states, and all the steady states have  the same eigenvalues (\cite{meiss2007differential}, Ch. 4). It follows that the interior steady state is asymptotically stable under the same assumptions.


Now we compute the Jacobians.
\begin{align*}
\hspace{-1cm}
&J_{\theta}F=\\
&\begin{pmatrix}\tilde{q}_{a}(1-2\theta_{a})(1-x_{a})-(1-\tilde{q}_{a})\theta_{v}(1-x_{v})-\mu, &\hspace{-0.5cm} (1-\tilde{q}_{a})(1-\theta_{a})(1-x_{v})\\
\hspace{-2.5cm}(1-\tilde{q}_{v})(1-\theta_{v})(1-x_{a}), &\hspace{-2.5cm} \tilde{q}_{v}(1-2\theta_{v})(1-x_{v})-(1-\tilde{q}_{v})\theta_{a}(1-x_{a})-\mu
\end{pmatrix}
\end{align*}
Since the two systems are smoothly equivalent, in the interior steady state this must have negative eigenvalues, and so is invertible. It follows that it has positive  determinant and negative trace. Moreover, as in the proof of Lemma \ref{relations}, since the off-diagonal elements are positive, the diagonal elements must be both negative. So, it follows 
 that the inverse $J_{\theta}F^{-1}$ has all negative entries.
\[
J_{x}F=\begin{pmatrix}-(1-\theta_{a})\theta_{a}\tilde{q}_{a} & -(1-\theta_{a})\theta_{v}(1-\tilde{q}_{a})\\
-(1-\theta_{v})\theta_{a}(1-\tilde{q}_{v}) & -(1-\theta_{v})\theta_{v}\tilde{q}_{v}
\end{pmatrix}
\]
The derivatives with respect to vaccination rates are collected in the matrix:
\[
\partial_x\theta=-J_{\theta}F^{-1}J_xF
\]

Since all the elements of the matrices $J_{\theta}F^{-1}$ and $J_xF$ are negative, it follows that 
$\partial_x\theta=-J_{\theta}F^{-1}J_xF$ has all negative elements, that is what we wanted to show.
\end{proof}

\subsubsection*{Proof of Lemma \ref{lemma:choices} (page \pageref{lemma:choices})}

\begin{proof} 

Define the system of implicit equations:
\begin{align*}
x_{a}&	=k\theta_{a}-d\\
x_{v}&	=k\theta_{v}
\end{align*}

And define the vector function $\Phi$ as:
\begin{align*}
\Phi_{a}(x_a,x_v,\theta_a,\theta_v)&	=x_{a}-k\theta_{a}+d\\
\Phi_{v}&	=x_{v}-k\theta_{v}
\end{align*}

To prove existence of an interior solution, we use the Poincaré-Miranda Theorem (\cite{kulpa1997poincare}). The conditions that have to be satisfied are:
\begin{align*}
\Phi_{a}(0,x_{v})&	=-k\theta_{a}+d\le 0\\
\Phi_{a}(1,x_{v})&	=1-k\theta_a+d\ge 0\\
\Phi_{v}(x_{a},0)&	=-k\theta_{v}\le 0\\
\Phi_{v}(x_{a},1)&	=1-k\theta_{v}\ge 0
\end{align*}

The solution is interior as long as these are satisfied.

The first needs just to be checked at the minimum of $\theta_{a}$, that is realized for $x_{v}=1$. In that case, we can solve explicitly for $\theta_{a}$, and we get: $\theta_{a}=1-\dfrac{\mu}{\tilde{q}_{a}}$,where in this case the upper bound on $\mu$ is $\hat{\mu}=\tilde{q}_{a}$, so this is always feasible. So the condition is: $d\le k\left(1-\dfrac{\mu}{\tilde{q}_{a}}\right)$.
The second and fourth need to be checked at the maxima of $\theta_a$ and $\theta_v$, that are realized, respectively, for $x_v=1$ and $x_a=1$. It can be calculated explicitly:
\begin{align*}
    \frac{(1-h)q\tilde{q}_{a}+h\mu}{(1-h)q(\tilde{q}_{a}-\mu)}&>k\\(1+d)\frac{(1-h)(1-q)\tilde{q}_{v}+h\mu}{(1-h)(1-q)(\tilde{q}_{v}-\mu)}&>k
\end{align*}
The second and $d\le k\left(1-\dfrac{\mu}{\tilde{q}_{a}}\right)$ are feasible only if:
\[
k\left(1-\dfrac{\mu}{\tilde{q}_{a}}\right)>k\frac{(1-h)(1-q)(\tilde{q}_{v}-\mu)}{(1-h)(1-q)\tilde{q}_{v}+h\mu}-1
\]
that is satisfied provided $k$ is low enough. So, we get two upper bounds $d\le \overline{d}$ and $k\le \overline{k}$. 

Now we prove uniqueness. The Jacobian of the system is:
\[
J_{x}\Phi=\begin{pmatrix}1-k\partial_{x_{a}}\theta_{a} & -k\partial_{x_{v}}\theta_{a}\\
-k\partial_{x_{a}}\theta_{v} & 1-k\partial_{x_{v}}\theta_{v}
\end{pmatrix}
\]
Given the signs computed in Lemma \ref{negative_theta}, this has positive diagonal and negative off-diagonal elements.

We have to study the sign of the determinant:
\begin{align*}
det J_{x}\Phi=&   1-k\partial_{x_{a}}\theta_{a}-k\partial_{x_{v}}\theta_{v}+k^{2}\left(\partial_{x_{a}}\theta_{a}\partial_{x_{v}}\theta_{v}-\partial_{x_{v}}\theta_{a}\partial_{x_{a}}\theta_{v}\right)
\end{align*}

Now the term \[\partial_{x_{a}}\theta_{a}\partial_{x_{v}}\theta_{v}-\partial_{x_{v}}\theta_{a}\partial_{x_{a}}\theta_{v}\] is the determinant of the matrix $\partial_x \theta$, that by the implicit function theorem is equal to:
\[
\begin{pmatrix}\partial_{x_{a}}\theta_{a} & \partial_{x_{v}}\theta_{a}\\
\partial_{x_{a}}\theta_{v} & \partial_{x_{v}}\theta_{v}
\end{pmatrix}=-J_{\theta}F^{-1}J_{x}F
\]

So: \[\partial_{x_{a}}\theta_{a}\partial_{x_{v}}\theta_{v}-\partial_{x_{v}}\theta_{a}\partial_{x_{a}}\theta_{v}=detJ_{\theta}F^{-1}det(-J_{x}F)\] 
The first determinant is positive by stability of the system; the second is: 
\begin{align*}
    det(-J_{x}F)&=detJ_{x}F=\theta_{a}\theta_{v}(1-\theta_{a})(1-\theta_{v})(\tilde{q}_{a}\tilde{q}_{v}-(1-\tilde{q}_{a})(1-\tilde{q}_{v}))\\
    &=\theta_{a}\theta_{v}(1-\theta_{a})(1-\theta_{v})h\ge0
\end{align*}
So, the determinant of $J_{x}$ is positive, and in particular the matrix is invertible, so there is locally a solution of the system. Moreover the determinant is positive, hence the matrix is positive definite, so the solution is unique and global thanks to the global implicit function theorem of \cite{gale1965jacobian}. 

Moreover, if $x_a^*\ge x_v^*$ it follows from Lemma \ref{relations} that $\rho_a\le \rho_v$ that implies $x_a^*\le x_v^*$, which is a contradiction: hence $x_a^*<x_v^*$.
\end{proof}

\subsubsection{Proof of Proposition \ref{prop:groups}}

\begin{proof}

We use the implicit function theorem. The derivatives of $\Phi$ with respect to homophily are:
\[
J_{h}\Phi=\begin{pmatrix}-k\partial_{h}\theta_{a}\\
-k\partial_{h}\theta_{v}
\end{pmatrix}
\]
So, the derivatives of the vaccination rates are:
\[
\partial_{h}x	=-J_{x}\Phi^{-1}J_{h}\Phi
	=\dfrac{1}{det J_{h}\Phi}\begin{pmatrix}1-k\partial_{x_{v}}\theta_{v} & k\partial_{x_{v}}\theta_{a}\\
k\partial_{x_{a}}\theta_{v} & 1-k\partial_{x_{a}}\theta_{a}
\end{pmatrix}\begin{pmatrix}k\partial_{h}\theta_{a}\\
k\partial_{h}\theta_{v}
\end{pmatrix}
\]
Now, $det J_{h}\Phi>0$. Moreover, $\partial_{h}\theta_{a}=\dfrac{1}{1-x_a}\partial_{h}\rho_{a}>0$ and $\partial_{h}\theta_{v}=\dfrac{1}{1-x_v}\partial_{h}\rho_{v}<0$ if $\mu<(1-x)^2/(1-x_a)$. So, it follows that if $\mu<(1-x)^2/(1-x_a)$, we have $\partial_h x_a>0$ and $\partial_h x_v<0$.    
\end{proof}

\subsubsection*{Proof of Proposition \ref{prop:endovax} (page \pageref{prop:endovax})}

\begin{proof}

\textbf{Part 1: steady state}

As clarified in the proof of Lemma \ref{lemma:choices}, to have an interior solution the conditions $d \le \overline{d}$ and $k\le \overline{k}$ must be satisfied. Moreover, $\overline{d}$ tends to zero if $k$ goes to zero. So, in the following, when we take the limits for $k$ and $d$ to zero, we always assume the conditions for an interior solutions are satisfied.

Consider first the limit $h\to 1$. The derivative of the vaccination rates becomes:
\begin{align*}
    \partial_{h}x&=\dfrac{1}{det J_x\Phi}\begin{pmatrix}1+k\dfrac{\mu}{(1-x_{a})^{2}} & 0\\
0 & 1+k\dfrac{\mu}{(1-x_{v})^{2}}
\end{pmatrix}\begin{pmatrix}k\partial_{h}\rho_{a}\dfrac{1}{1-x_{a}}\\
k\partial_{h}\rho_{v}\dfrac{1}{1-x_{v}}
\end{pmatrix}\\
&=\begin{pmatrix}k\left(1+k\dfrac{\mu}{(1-x_{v})^{2}}\right)^{-1}\dfrac{1}{1-x_{a}}\partial_{h}\rho_{a}\\
k\left(1+k\dfrac{\mu}{(1-x_{a})^{2}}\right)^{-1}\dfrac{1}{1-x_{v}}\partial_{h}\rho_{v}
\end{pmatrix}
\end{align*}
where $J_x\Phi=\left(1+k\dfrac{\mu}{(1-x_{a})^{2}}\right)\left(1+k\dfrac{\mu}{(1-x_{v})^{2}}\right)$.
So we get:
\[
q\partial_{h}x_{a}+(1-q)\partial_{h}x_{v}>k\left(1+k\dfrac{\mu}{(1-x_{a})^{2}}\right)^{-1}\dfrac{1}{1-x_{v}}\partial_{h}\rho
\]
and so, the total derivative satisfies:
\[
d_{h}\rho<\partial_{h}\rho\left(1-\dfrac{k\dfrac{1}{1-x_{v}}}{1+k\dfrac{\mu}{(1-x_{a})^{2}}}\right)
\]
The smallest level of $d$ to have an interior solution is $d>k/(1-\mu)$.
In the limit for $d\to k/(1-\mu)$ and $k\to 0$, we have that the term 
$\dfrac{k\dfrac{1}{1-x_{v}}}{1+k\dfrac{\mu}{(1-x_{a})^{2}}}$ goes to zero, and so the derivative has the same sign as $\partial_h\rho$.

Now focus on the limit $h\to 0$. 

In this limit the derivative of vaccination rates are:
\begin{align*}
    \partial_{h}x&=\dfrac{1}{1+\dfrac{k\mu}{(1-x)^{2}}}\begin{pmatrix}1+k\dfrac{(1-q)\mu}{(1-x)^{2}} & -k\dfrac{(1-q)\mu}{(1-x)^{2}}\\
-k\dfrac{q\mu}{(1-x)^{2}} & 1+k\dfrac{q\mu}{(1-x)^{2}}
\end{pmatrix}\begin{pmatrix}k\partial_{h}\rho_{a}\dfrac{1}{1-x_{a}}\\
k\partial_{h}\rho_{v}\dfrac{1}{1-x_{v}}
\end{pmatrix}\\
&=\dfrac{1}{1+\dfrac{k\mu}{(1-x)^{2}}}\begin{pmatrix}k\partial_{h}\rho_{a}\dfrac{1}{1-x_{a}}+k^{2}\dfrac{(1-q)\mu}{(1-x)^{2}}\left(\partial_{h}\rho_{a}\dfrac{1}{1-x_{a}}-\partial_{h}\rho_{v}\dfrac{1}{1-x_{v}}\right)\\
k\partial_{h}\rho_{v}\dfrac{1}{1-x_{v}}+k^{2}\dfrac{q\mu}{(1-x)^{2}}\left(\partial_{h}\rho_{v}\dfrac{1}{1-x_{v}}-\partial_{h}\rho_{a}\dfrac{1}{1-x_{a}}\right)
\end{pmatrix}
\end{align*}
When aggregating, we get a simplification:
\begin{align*}
    q\partial_{h}x_{a}+(1-q)\partial_{h}x_{v}&	=\dfrac{k}{1+\dfrac{k\mu}{(1-x)^{2}}}\left(q\partial_{h}\rho_{a}\dfrac{1}{1-x_{a}}+(1-q)\partial_{h}\rho_{v}\dfrac{1}{1-x_{v}}\right)
\end{align*}

We can compute this analytically:
\[
q\partial_{h}\rho_{a}\dfrac{1}{1-x_{a}}+(1-q)\partial_{h}\rho_{v}\dfrac{1}{1-x_{v}}=\frac{\mu ^2 (1-q) q (x_v-x_a)^2}{1-x)^4}
\]
So, the total derivative is:
\begin{align*}
d_h\rho=&\partial_h\rho-(q\partial_{h}x_{a}+(1-q)\partial_{h}x_{v})\\    =&\frac{\mu ^2 (1-q) q (x_v-x_a)^2}{1-x)^2}\left(\dfrac{1}{(1-x)^2}-\dfrac{k}{1+\dfrac{k\mu}{(1-x)^{2}}}\right)
\end{align*}
and again, in the limit for $d\to \overline{d}$ and $k\to 0$, we get that the second term disappears and the derivative is positive.

\textbf{Part 2: cumulative infection}

Focus on the limit for $h\to 1$.

Remember that in this limit we have $\partial_{x_a}\rho_a=-1$, $\partial_{x_v}\rho_a=0$, and analogously for $\rho_v$. So:
\[
d_{x_a}=CI\partial_{x_{a}}CI-\partial_{\rho_{a}}CI=\dfrac{q}{(1-x_{a}-\mu+r)^{2}}
\]
\[
d_{x_v}=CI\partial_{x_{v}}CI-\partial_{\rho_{v}}CI=\dfrac{1-q}{(1-x_{v}-\mu+r)^{2}}
\]

Moreover, notice that, since $\partial_{h}\rho_{a}>0$:
\[
\partial_hx_a=\dfrac{k\dfrac{1}{1-x_{a}}}{1+k\dfrac{\mu}{(1-x_{v})^{2}}}\partial_{h}\rho_{a}>
\dfrac{k\dfrac{1}{1-x_{v}}}{1+k\dfrac{\mu}{(1-x_{a})^{2}}}\partial_{h}\rho_{a}
\]

Using these, the indirect effect is:
\begin{align*}
&d_{x_{a}}CI\partial_{h}x_{a}+d_{x_{v}}CI\partial_{h}x_{v}\\
&=
\left(\partial_{x_{a}}CI-\partial_{\rho_{a}}CI\right)\partial_{h}x_{a}+\left(\partial_{x_{v}}CI-\partial_{\rho_{v}}CI\right)\partial_{h}x_{v}	\\
&=
\dfrac{q}{(1-x_{a}-\mu+r)^{2}}\partial_{h}x_{a}\dfrac{1-q}{(1-x_{v}-\mu+r)^{2}}\partial_{h}x_{v}	\\
&>k\dfrac{\dfrac{1}{1-x_{v}}}{1+k\dfrac{\mu}{(1-x_{a})^{2}}}\left(\dfrac{q}{(1-x_{a}-\mu+r)^{2}}\partial_{h}\rho_{a}+\dfrac{1-q}{(1-x_{v}-\mu+r)^{2}}\partial_{h}\rho_{v}\right)
\end{align*}

Using the results of the proof of Proposition \ref{Cumulative}, we obtain that the total effect satisfies:
\[
d_hCI>\left(-\dfrac{2}{(1-x_{a}-\mu+r)^{2}}+\dfrac{1}{(1-x_{v}-\mu+r)^{2}}\dfrac{k\dfrac{1}{1-x_{v}}}{1+k\dfrac{\mu}{(1-x_{a})^{2}}}\right)\partial_{h}\rho
\]

Now taking the limit for  $k\to d/(1-\mu)=0$ and $d\to0$, we obtain that the second term disappears, and the derivative is positive.

Now focus on $h\to 0$.
\[
d_{x_a}CI=CI\partial_{x_{a}}CI-\partial_{\rho_{a}}CI=\frac{q}{(-r+\mu+x-1)^{2}}
\]
\[
d_{x_v}CI=CI\partial_{x_{v}}CI-\partial_{\rho_{v}}CI=\frac{1-q}{(-r+\mu+x-1)^{2}}
\]
So, there is a simplification:
\begin{align*}
&d_{x_{a}}CI\partial_{h}x_{a}+d_{x_{v}}CI\partial_{h}x_{v}\\
&=
\left(\partial_{x_{a}}CI-\partial_{\rho_{a}}CI\right)\partial_{h}x_{a}+\left(\partial_{x_{v}}CI-\partial_{\rho_{v}}CI\right)\partial_{h}x_{v}	\\
&=
\dfrac{q}{(-\beta+\mu+x-1)}\partial_{h}x_{a}+\dfrac{1-q}{(-r+\mu+x-1)^{2}}\partial_{h}x_{v}	\\
&=\dfrac{1}{(-r+\mu+x-1)^2}\left(q\partial_{h}x_{a}+(1-q)\partial_{h}x_{v}\right)
\end{align*}
Now using the analytical results obtained for the steady state case, we get that the indirect effect is:
\[
d_{x_{a}}CI\partial_{h}x_{a}+d_{x_{v}}CI\partial_{h}x_{v}=\dfrac{1}{(-\beta+\mu+x-1)^2}\frac{\mu ^2 (1-q) q (x_v-x_a)^2}{(1-x)^4}
\]
Adding the direct effect computed in Proposition \ref{Cumulative}, we get: that the derivative is negative if and only if:
\[
\frac{(1-x)^2 (3(1-x)+2
   \beta -2 \mu )}{1-x+r}-\frac{k \mu
   }{\frac{k \mu }{(1-x)^2}+1}>0
\]
and the second term goes to zero in the limit for $d$ and $k$ small, so the derivative is indeed negative.
\end{proof}

\subsubsection*{Proof of Proposition \ref{prop:peer} (page \pageref{prop:peer})}

\label{app:peer}

\begin{proof}

\textbf{Part 1}

First we summarize the complete equilibrium characterization, then we prove it.

\textbf{Equilibria}: For all values of the parameters, there is an equilibrium with $x_a^*=x_v^*=0$.
For $k\le 1$, this is the unique equilibrium.
For $h<1-\frac{k-1}{kq}$ and $d<\frac{(k-1)(1-hk)}{k q(1-h)}$, there is also an equilibrium where both are interior: $x_a^*,x_v^*\in (0,1)$. If $h\ge 1-\frac{k-1}{kq}$ and, either $d>k(1-\tilde{q}_a)$, or $k\tilde{q}_a>1$, then there is an equilibrium with $x_a^*=0$ and $x_v^*=1$.
If $k\ge d+1$, there is an equilibrium with $x_a^*=x_v^*=1$. If $1+d>k\ge 1/\tilde{q}_v$, $k(1-\tilde{q}_a)> d$ and $k\tilde{q}_a<1$, there is an equilibrium where $x_v^*=1$ and $x_a^*$ is interior.

The equilibria (including corner ones) are the solutions of the following equations:
\begin{align}
x_a=&\min\{1,\max\{0,k(\tilde{q}_a x_a+(1-\tilde{q}_a)x_v)-d\}\},\nonumber \\
x_v=&\min\{1,\max\{k(\tilde{q}_v x_v+(1-\tilde{q}_v)x_a)\}\}.
\end{align}

From the equations, it follows immediately that $x_a^*=x_v^*=0$ is always an equilibrium. In addition, there might be interior or corner equilibria.

There is a corner equilibrium in which $x_a^*=x_v^*=1$ if $x_a^*=1\le k(\tilde{q}_a+1-\tilde{q}_a)-d=k-d$, that is $k\ge 1+d$. This condition immediately implies $x_a=1$, and this, in turn, implies $k(\tilde{q}_v+1-\tilde{q}_v)>1$, that implies that also $x_v=1$.

There is a corner equilibrium in which $x_a^*=0$ and $x_v^*=1$ if $k\tilde{q}_v\ge 1$ and $k(1-\tilde{q}_a)\le d$. If $1+d>k\ge 1/\tilde{q}_v$, $k(1-\tilde{q}_a)> d$ and $k\tilde{q}_a<1$, there is an equilibrium where $x_v^*=1$ and $x_a$ interior, equal to:
\[
x_a^*=\dfrac{k(1-\tilde{q}_a)- d}{1-k\tilde{q}_a}
\]

If $k<1$, supposing that the system \eqref{vax_peer} has interior solution, the linear system can be written as:
\[
\left(I-k\left(\begin{array}{cc}
\tilde{q}_a & 1-\tilde{q}_a\\
1-\tilde{q}_v & \tilde{q}_v
\end{array}\right)\right)\left(\begin{array}{c}
x_a\\
x_v
\end{array}\right)=\left(\begin{array}{c}
-d\\
0
\end{array}\right).
\]
where $I$ is the $2\times 2$ identity matrix. Since the matrix $\left(\begin{array}{cc}
\tilde{q}_a & 1-\tilde{q}_a\\
1-\tilde{q}_v & \tilde{q}_v
\end{array}\right)$ is stochastic, it follows that the maximum eigenvalue is 1. So, using standard results, if $k<1$ the matrix $\left(I-k\left(\begin{array}{cc}
\tilde{q}_a & 1-\tilde{q}_a\\
1-\tilde{q}_v & \tilde{q}_v
\end{array}\right)\right)$ is invertible and has positive inverse. So, it follows that the unique interior equilibrium would satisfie:
\[
\left(\begin{array}{c}
x_a^*\\
x_v^*
\end{array}\right)=\left(I-k\left(\begin{array}{cc}
\tilde{q}_a & 1-\tilde{q}_a\\
1-\tilde{q}_v & \tilde{q}_v
\end{array}\right)\right)^{-1}\left(\begin{array}{c}
-d\\
0
\end{array}\right).
\]
but, since the inverse is positive, it follows that $x_a^*,x_v^*<0$, which is not feasible. So, for $k<1$, there is no interior solution. Moreover, it is also not possible to have a partially interior solution with $x_v^*>0$ but $x_a^*=0$. Indeed, if $x_a^*=0$, then $x_v^*$ satisfies the equation:
\[
x_v=k\tilde{q}_vx_v
\]
and since $k<1$ then $k\tilde{q}_v<1$, so the only solution is $x_v^*=0$. A similar reasoning works for $x_a^*$.

If $k=1$, from the equations we obtain $x_v-x_a=0$ and $x_a-x_v=-d/(1-\tilde{q}_a)$, that is not feasible. So, again the only solution is for both rates to be zero.

So, we focus on the case of $k>1$. With a calculation, we find that the interior solution is:
\begin{align*}
x_a^*=&\frac{d (1-k (1-(1-h) q))}{(k-1) (1-h k)}, \quad
x_v^*= \frac{d (1-h) k q}{(k-1) (1-hk)}.
\end{align*}
and, in this interior solution, $x_v^*-x_a^*=\frac{d}{1-hk}$, so $x_v^*>x_a^*$, as expected.
This is feasible if $x_a^*>0$, and if $x_v^*<1$. It can be checked that this is the case if $1<k<\frac{1}{\tilde{q}_v}$ and $d<\frac{(k-1)(1-hk)}{k q(1-h)}$. The first condition implies $hk<1$, so in particular the second bound for $d$ is always nonnegative.


\textbf{Part 2}

This part simply follows from the calculation of the derivatives:
\begin{align*}
\partial_h x_a&=-\frac{d k (1-q)}{(h k-1)^2}, \quad 
\partial_h x_v=\frac{d k q}{(h k-1)^2},
\end{align*}

\textbf{Part 3}

For $h\to 1$, given the characterization in the proof of Part 1, regardless of the other parameters, the only remaining equilibria are corner equilibria where vaccinations levels are either zero or one. So, the effect of increasing homophily is zero: $\partial_h x_a=\partial_h x_v=0$. As a consequence, the comparative statics with respect to homophily satisfy $d_h CI>0$, and $d_h \rho^{SS}<0$, as in Proposition \ref{prop:endovax}.
\end{proof}

\clearpage

\bibliographystyle{chicago}
\bibliography{biblio}

\end{document}